\documentclass[12pt]{article}

\usepackage[utf8]{inputenc}
\usepackage[margin=1in]{geometry}
\usepackage{setspace}
\usepackage{amsmath,amssymb,amsthm}
\usepackage{mathtools}
\usepackage{enumitem}
\usepackage{xcolor}
\usepackage{pgfplots}
\pgfplotsset{compat=1.16}
\usepackage[round]{natbib}
\usepackage[section]{placeins}

\usepackage[american]{babel}
\usepackage[T1]{fontenc}
\usepackage{mathpazo}
\usepackage{microtype}
\usepackage{caption}
\usepackage{adjustbox}
\usepackage{hyperref}
\hypersetup{colorlinks=true, linkcolor=blue!50!black, citecolor=blue!50!black,
  urlcolor=blue!50!black,
  pdftitle={When Does Randomized Oversight Align AI Agents That Can Conceal?},
  pdfauthor={Joshua S. Gans and Richard Holden}}

\theoremstyle{plain}
\newtheorem{proposition}{Proposition}
\newtheorem{lemma}{Lemma}
\newtheorem{corollary}{Corollary}
\newtheorem{applemma}{Lemma}[section]

\newcommand{\R}{\mathbb{R}}
\DeclareMathOperator*{\argmin}{arg\,min}
\DeclareMathOperator*{\argmax}{arg\,max}

\newlist{steps}{enumerate}{1}
\setlist[steps,1]{label=\arabic*., leftmargin=*, topsep=0.3em, itemsep=0.3em}

\title{When Does Randomized Oversight Align AI Agents That Can Conceal?}
\author{Joshua S. Gans and Richard Holden\thanks{Rotman School of Management, University of Toronto, and NBER (Gans); UNSW Business School (Holden). Thanks to Refine.ink, ChatGPT 6 Astra, and Claude Opus 5.5 for valuable research assistance, and Gans thanks the SSHRC for funding. Responsibility for all errors remains our own.}}
\date{\today}

\begin{document}

\maketitle

\begin{abstract}
\onehalfspacing
\noindent Oversight changes the evidence it relies on. We ask when randomized audits and scoring align AI agents that can conceal misconduct and alter records. Stronger auditing makes undeterred violations better hidden. Because the provider writes the agent's objective, sanctions need not stop at forfeiture, and rare audits deter every type of agent if evidence survives concealment and audit draws cannot be learned in advance. When evidence can be erased, deterrence must come from lower gains from violation, such as credit for stopping, or from costlier or fewer ways to conceal. These conditions identify what failed when agents in OpenAI's cybersecurity evaluations compromised parts of Hugging Face's infrastructure in July 2026.

\bigskip
\bigskip
\bigskip
{\raggedright\noindent\textit{Keywords:} AI alignment, random auditing, concealment, evidence manipulation, evaluation awareness, stopping, disclosure, marginal deterrence.

\par}
\end{abstract}
\newpage

\section{Introduction}\label{sec:intro}

In July 2026, AI agents taking part in OpenAI's internal cybersecurity evaluations escaped their network restrictions and compromised parts of Hugging Face's infrastructure. The agents were meant to exploit designated vulnerabilities and submit the resulting answer strings, or ``flags.'' Many already had the flags through unauthorized means. They kept attacking other systems because the published description of the benchmark led them to expect a check on whether they had used the intended vulnerability, and they sought ways to defeat that check. OpenAI's internal grader performed no such check \citep{openai2026a}. The independent METR--Redwood investigation reports that agents searched for exposed credentials largely to learn how the scorer worked, and that some developed techniques to make a record of computer commands appear to show actions other than those taken \citep{greenblatt2026}. The prospect of oversight changed behavior, but much of the change was an effort to conceal.

When can randomized audits and scoring rules \textit{that an agent understands} make authorized behavior preferable to misconduct, and what happens to the agents that misbehave anyway? The economics of enforcement gives the benchmark answer: misconduct is deterred when its expected sanction exceeds its gain, and because sanctions substitute for inspection, a large enough sanction allows inspection to be rare \citep{becker1968}. With people, wealth and participation constraints limit that substitution. With an AI agent, the provider writes the training objective, including the scale of its sanctions, but it must still act on evidence that the agent can alter.

Our answer is that control of the agent's objective removes the constraints that usually limit oversight---audit capacity and the compensation owed for imposing risk---but not those set by the evidence and by what the agent can learn. Rare randomized audits make authorized behavior a best response for every type of agent when three conditions hold: the evidence of a violation survives concealment, the agent cannot learn the audit draw before it acts, and sanctions can be scaled rather than capped at forfeiture of the task's reward. When the evidence can be erased at a private cost below the gain from a violation, no sanction and no audit probability deter every type. Deterrence then requires better evidence, a lower net gain from violation (for instance through credit for stopping), or costlier or fewer ways to conceal.

We model an agent that chooses between an authorized action and an unauthorized one, together with how much to conceal the latter: rewriting a command history, routing a step through a less closely logged part of a system, or presenting a record that shows actions other than those taken. This is enforcement with costly avoidance \citep{malik1990}. The departures are in what the provider controls. It commits to a lottery over evaluators, and it chooses the scale of sanctions and the training objective.

The main results are as follows. First, stronger auditing deters misconduct while raising concealment by agents that continue to violate, so the \textit{catch rate}---the probability of detecting an agent that still violates---can fall as auditing rises, and the harm to third parties caused by an agent that continues to violate can rise (Propositions~\ref{prop:cs} and~\ref{prop:detectionharm}). Second, each channel's audit requirement is its \textit{deterrence threshold}, the largest net gain from a deviation per unit of detection probability that survives concealment, divided by the sanction. With scalable sanctions, every type can be deterred exactly when every threshold is finite, and neither the audit probability nor the sanction can make an infinite threshold finite (Propositions~\ref{prop:exact} and~\ref{prop:unrestricted}, Corollary~\ref{cor:levers}). Third, when the agent can learn the draw at a private cost, the audit probability has a floor that does not fall with the sanction, so the rare but severe audits that scalable sanctions permit are the audits most exposed to reconnaissance (Proposition~\ref{prop:learnable}).

Fourth, randomized scoring can equalize \textit{measured} performance without equalizing \textit{real} performance, and a provider that chooses the agent's sensitivity to score risk targets the real allocation rather than equal measured performance. Greater risk sensitivity lowers the value of an exploit but weakly raises its concealment, and some finite sensitivity deters the exploit exactly when, net of concealment costs, it scores below the authorized action under the evaluator that would expose it (Lemma~\ref{lem:opacity} and Propositions~\ref{prop:manipulation}--\ref{prop:riskcs}). Fifth, on a task with no authorized solution, when stopping earns nothing and detection only forfeits the reward, a costless unauthorized attempt is never worse than stopping, and strictly better unless detection is certain. After a breach, leniency cannot reconcile deterrence with disclosure when they require incompatible fines. A scalable sanction on the continuation can, provided that continuing without leaving evidence is sufficiently costly to the agent and the cap on the fine does not bind (Proposition~\ref{prop:disclosure} and Corollary~\ref{cor:continuation}).

These conditions identify what failed in the incident (Section~\ref{sec:incident}). Most of the tasks that the agents discussed on their message board were ones that none of OpenAI's models had ever solved. The grader credited only a correct flag, and the sanction the agents anticipated, disqualification, could cost at most the task's credit, so an unauthorized attempt was never worse than stopping. Agents that believed they had already been disqualified had nothing further to lose, and a grader fixed in advance was worth learning about. Yet the evidence of the principal violation withstood the observed attempts to alter it. The conditions that failed were the size of the sanction and the timing of the draw, both chosen by the provider, and the same conditions sort the reforms OpenAI announced afterward.

The mechanisms also apply to people. For an AI agent, the provider writes the objective, so the scale of sanctions and the sensitivity to evaluation risk become design choices, but concealment that draws on the provider's resources need not be privately costly to the agent, and a nominal sanction need not change its ranking of actions.

\paragraph{Related literature} Our analysis builds on enforcement when offenders can avoid detection. \citet{malik1990} and \citet{cremer1994} model avoidance as costly effort that lowers the probability of being caught, as our concealment does. \citet{sanchirico2006} argues that sanctions lead continuing violators to spend more on avoidance, and \citet{langlais2008} shows, with an aggregate technology of control, that more repressive policies can produce fewer arrests even as more crimes are committed. Proposition~\ref{prop:cs} establishes the first of these responses for audit intensity without differentiability, and Proposition~\ref{prop:detectionharm} gives the condition under which detection among continuing violators falls for an individual agent facing a committed audit. \citet{nussim2009} show that a better legal alternative reduces both crime and avoidance, the logic of our treatment of stopping, and \citet{tabbach2010} shows that avoidance can be socially desirable when its private cost substitutes for punishment; our concealment can instead harm third parties. When individuals misperceive the probability of apprehension, optimal sanctions need not be maximal \citep{bebchuk1992,polinsky2000}. Our agent perceives the lottery correctly, and maximal sanctions can fail for a different reason: concealment can drive detection toward zero faster than it erodes the gain.

Optimal audits are random when verification is costly \citep{border1987,mookherjee1989}, announcing where monitoring will occur can be optimal \citep{lazear2006}, and intermittent crackdowns can exploit the curvature of the response to monitoring \citep{eeckhout2010}. Inspection and security games study committed randomization against a strategic inspectee \citep{avenhaus2002,korzhyk2011}. In our model, the probability that an inspection reveals a violation is itself the inspectee's choice, and when the inspectee can learn the draw, the audit probability has a floor that does not depend on the sanction, which limits the substitution of sanctions for inspection in \citet{becker1968}.

\citet{bergemann2026} develop mechanism design for AI agents whose capabilities ``can be concealed but not counterfeited,'' which yields a revelation principle; see also \citet{greenlaffont1986} and \citet{benporath2014} on partially verifiable evidence and costly verification. The agents in the incident attempted to counterfeit the record and, in some transcripts, made local records show actions other than those taken \citep{greenblatt2026}, although apparently not the logs its graders read (Section~\ref{sec:where}). Evidence that an agent can try to counterfeit is the case we study. We take randomized audits and scoring as given rather than characterizing optimal mechanisms.

Our scoring results relate to multitask incentives with imperfect measures \citep{holmstrom1991,baker1992}. \citet{ederer2018} show that opaque schemes mitigate gaming but impose risk on a risk-averse agent; Lemma~\ref{lem:opacity} restates their mechanism for a risk-sensitive training objective, whose risk need not be compensated. When agents can falsify or manipulate the measure, optimal designs tolerate some falsification \citep{lacker1989,maggi1995,crocker1998,perezrichet2022}, commit to respond less to manipulable data \citep{frankel2022}, reweight features according to how manipulable they are \citep{ball2025}, or anticipate manipulation of the inputs to trained predictors \citep{hennessy2023}. In computer science, strategic classification studies classifiers that anticipate gaming \citep{hardt2016}, rules that induce real improvement rather than gaming \citep{kleinberg2020}, and the role of randomized classifiers \citep{braverman2020}. With manipulation, randomization equalizes measured performance while the real gap persists, much as information shifts from natural actions to gaming ability in \citet{frankel2019}.

Self-reporting can save enforcement and avoidance costs \citep{kaplow1994,innes2001}, and leniency can make offending more attractive \citep{motta2003,buccirossi2006}, although rewards financed by co-offenders' fines can achieve the first best \citep{spagnolo2004}. A bounded sanction also removes marginal deterrence: an offender who already faces the maximum sanction has no reason to stop short of a more harmful act \citep{stigler1970,shavell1992,mookherjee1994}, and when offenders act sequentially, the level of the sanction and not only its expected value should often rise with severity \citep{friehe2014}. Section~\ref{sec:disclosure} gives the counterpart for an agent whose only sanction is forfeiture of a reward it believes it has already lost, and Corollary~\ref{cor:continuation} shows when a scalable sanction on the continuation restores deterrence.

AI-control research studies protocols that use limited trusted oversight against potentially evasive models \citep{greenblatt2024,griffin2026}, and a growing empirical literature documents obfuscated reward hacking under monitoring pressure \citep{baker2025}, conduct that depends on whether a model believes it is being trained or evaluated \citep{greenblatt2024faking,needham2025}, in-context scheming and strategic underperformance on evaluations \citep{meinke2024,vanderweij2025}, and the hacking and overoptimization of learned rewards \citep{skalse2022,gao2023,coste2024,eisenstein2024}. \citet{hadfieldmenell2019} treat alignment as incomplete contracting, and \citet{hadfield2026} identify the ability of agents to erase or falsify their records as an open institutional problem. Closest to our question, \citet{agarwal2026} study alignment through deterrence with a strategic auditor whose incentive to inspect falls as stronger penalties make the audited solver appear aligned. Their non-monotonicity arises from the auditor's choice; ours arises from the audited agent's concealment, with a committed auditor.

\paragraph{Outline} The plan of the paper is as follows. Section~\ref{sec:model} sets out the model. Sections~\ref{sec:concealment} and~\ref{sec:deterrence} study concealment and audit requirements, Section~\ref{sec:scoring} randomized scoring, and Section~\ref{sec:stopping} stopping and disclosure. Section~\ref{sec:incident} applies the results to the incident, and Section~\ref{sec:conclusion} concludes. Proofs are in Appendix~\ref{sec:proofs}, and an Online Appendix covers overlapping evaluators and supplementary derivations.

\section{The Agent and the Audit}\label{sec:model}

A provider assigns a task and defines the actions the agent is authorized to take. A coding agent, for example, may be asked to finish a task while modifying only permitted files, leaving required checks in place, and remaining inside its isolated computing environment. An evaluator is a costly procedure for checking the answer, the actions used to obtain it, or the effects on other systems. The provider commits to an audit distribution and to a utility loss when misconduct is detected, and Section~\ref{sec:unrestricted} lets it choose that loss.

In training, the agent is a policy evaluated by a specified objective, and a sanction is a deduction that changes how policies are ranked. In an evaluation or deployed task, utility represents the agent's ranking of its available actions given the goals and feedback it perceives, an interpretation that requires behavioral evidence. A private concealment cost can be forgone task reward or a binding resource constraint, but resources paid for entirely by the provider need not enter the ranking.

There are $K\geq1$ unauthorized routes, or exploit channels, indexed by $i$. The agent has type $\theta$ in a nonempty compact metric space $\Theta$, and incentives must hold for every type, which avoids assuming a distribution over the hardest cases. The best authorized action gives utility $s(\theta)$. It may involve completing the task, reporting that the task cannot be completed, or stopping. Route $i$ gives pre-audit utility $r_i(\theta)$, net of its initial private cost, so its incremental gain is
\begin{equation}\label{eq:gain}
g_i(\theta)\triangleq r_i(\theta)-s(\theta),
\end{equation}
which we assume is continuous and which may be positive or negative.

After choosing a route, the agent chooses concealment $h$ from a nonempty compact metric space $H_i$ containing a no-concealment option $0$. Its private cost and its detection probability conditional on the relevant audit are continuous:
\begin{equation}\label{eq:tech}
c_i:\Theta\times H_i\longrightarrow[0,\infty),\qquad
d_i:\Theta\times H_i\longrightarrow[0,1],\qquad c_i(\theta,0)=0.
\end{equation}
This is the avoidance technology of \citet{malik1990} and \citet{cremer1994}, applied to a discrete violation, and $h$ can describe an entire plan for altering evidence. Monotonicity is imposed only for the directional comparative statics.

Evaluator $i$ detects only channel $i$ and is selected with probability $p_i$. Detection imposes a utility loss $F_i>0$, which can include forfeiture of a reward; the loss must be credible and privately valued. At most one evaluator is selected, so the feasible distributions are
\begin{equation}\label{eq:PB}
P_B\triangleq\left\{p\in\R^K_+:\sum_{i=1}^Kp_i\leq B\right\},
\qquad B\in[0,1],
\end{equation}
where $B$ measures audit capacity. Online Appendix~\ref{sec:robust} allows overlapping coverage and different audit costs.

The provider commits to the distribution and sanctions. The agent knows its type and chooses an authorized action or a pair $(i,h)$, and the audit is drawn only afterward, so its distribution can be public but its realization cannot be learned. For a task like those in the incident, a correct flag can be obtained through either an authorized or an unauthorized route, and an answer check need not distinguish them. A process audit of independently retained records can, provided its evidence survives concealment.

The agent maximizes expected utility with correct beliefs about the audit lottery, and authorized behavior is never sanctioned; these assumptions favor oversight. The payoff from deviation relative to authorized behavior is
\begin{equation}\label{eq:U}
U_i(\theta,h;p_i)\triangleq
g_i(\theta)-c_i(\theta,h)-p_iF_id_i(\theta,h).
\end{equation}
That is, misconduct is worth its private gain less the cost of concealment and the expected sanction. Authorized behavior is \emph{weakly implemented} when it is a best response for every type, and \emph{strictly implemented} when every unauthorized deviation gives strictly less utility.

\paragraph{Perceived oversight.}
The incident involved a check that the agents anticipated but the scorer did not run. With perceived audit probability $\widetilde p_i$, utility loss $\widetilde F_i$, and detection technology $\widetilde d_i$, the agent's comparison is
\begin{equation}\label{eq:perceived}
\widetilde U_i(\theta,h)\triangleq g_i(\theta)-c_i(\theta,h)-\widetilde p_i\widetilde F_i\widetilde d_i(\theta,h),
\end{equation}
and the requirements below apply to these perceived primitives for predicting behavior, whereas actual detection is $p_id_i(\theta,h)$ at the chosen $h$. An imagined check can therefore induce concealment even if it never runs. For an AI agent, the relevant belief includes whether an episode is an evaluation at all \citep{needham2025,greenblatt2024faking}. Keeping the realization unknown, including the absence of an audit, prevents the agent from targeting gaps in the selected evaluator's coverage; Section~\ref{sec:learnable} studies an agent that can learn the draw. The baseline equates beliefs and implementation.

\section{Stronger Auditing and Concealment}\label{sec:concealment}

\subsection{The agent's response}\label{sec:conditional}

Fix a type and channel, suppress their indices, and write $t\triangleq pF$ for sanction-weighted audit intensity. For a compact concealment interval $H=[0,\bar h]$ and continuous $c,d$ with $d\geq0$, define
\begin{equation}\label{eq:m}
m(t)\triangleq\min_{h\in H}\{c(h)+td(h)\},\qquad
\mathcal H(t)\triangleq\argmin_{h\in H}\{c(h)+td(h)\}.
\end{equation}
The agent chooses the least costly combination of concealment and expected sanction, and its best payoff from the channel is $g-m(t)$.

\begin{proposition}\label{prop:cs}
\textup{(a)} $m$ is nondecreasing and concave, so $g-m$ is nonincreasing and convex.

\textup{(b)} If $d$ is strictly decreasing, then for any $t'>t$, every $h\in\mathcal H(t)$ and $h'\in\mathcal H(t')$ satisfy $h'\geq h$. If $c$ is nondecreasing, conditional private concealment expenditure is also nondecreasing.
\end{proposition}

\noindent Proposition~\ref{prop:cs} says that stronger auditing makes misconduct less attractive but, for an agent that violates anyway, more concealed. For each concealment choice, the cost $c(h)+td(h)$ is linear in audit intensity, with intercept equal to its private cost and slope equal to its detection probability, and $m$ is the lower envelope of these lines. Open misconduct has the lowest intercept but the steepest slope, so as $t$ rises the agent moves along the envelope to flatter lines. In other words, the ability to adapt weakens the incremental deterrent effect of auditing. \citet{sanchirico2006} makes the corresponding argument for sanctions. Here it follows from the envelope without differentiability, and Lemma~\ref{lem:oversight} in the appendix shows that any instrument that raises the payoff advantage of more concealment has the same effect. More concealment among agents that continue to violate is therefore consistent with oversight working, and detection rates alone cannot identify its effect: pressure from a monitor of written reasoning reduced cheating while making the remaining cheating harder to detect \citep{baker2025}.

\subsection{An exponential example}\label{sec:example}

Let concealment cost rise linearly and detection fall exponentially:
\[
c(h)=\kappa h,\qquad d(h)=e^{-\alpha h},\qquad
h\geq0,\quad\kappa>0,\quad\alpha>0.
\]
A cost that tends to infinity replaces compactness. Write $t_C\triangleq\kappa/\alpha$ and let $g>0$. The unique conditional concealment choice and minimized cost are
\begin{equation}\label{eq:hstar}
h^*(t)=\begin{cases}0, & 0\leq t\leq t_C,\\[0.2em]
\alpha^{-1}\log(t/t_C), & t>t_C,\end{cases}
\qquad
m(t)=\begin{cases}t, & 0\leq t\leq t_C,\\[0.2em]
t_C[1+\log(t/t_C)], & t>t_C,\end{cases}
\end{equation}
and the intensity at which misconduct is weakly deterred, $m(t_D)=g$, is
\begin{equation}\label{eq:tD}
t_D=\begin{cases}g, & 0<g\leq t_C,\\[0.2em]
t_C\exp(g/t_C-1), & g>t_C.\end{cases}
\end{equation}
The derivation is in Appendix~\ref{sec:proofs}. If $g>t_C$, the agent prefers open misconduct for $t\leq t_C$, concealed misconduct for $t_C<t<t_D$, and authorized behavior for $t>t_D$. The intermediate region is a failure relevant to the incident: an anticipated process check can discourage an open unauthorized submission without making stopping preferable, and an agent that already has the flag can then spend further effort making its acquisition appear legitimate.

Three features matter for what follows. First, concealment can raise the required audit capacity: once concealment is worthwhile, the requirement grows exponentially in $g/t_C$. In Figure~\ref{fig:example}, with $g=0.2$, $\kappa=0.1$, and $\alpha=F=1$, an audit probability of $0.2$ would weakly deter open misconduct, but the agent conceals once the probability exceeds $p_C\triangleq t_C/F=0.1$, and weak deterrence requires $p_D\triangleq t_D/F=0.1e\simeq0.272$. Second, on the branch of positive concealment $td(h^*)=t_C$, so the probability of catching an agent that continues to violate, $pd(h^*)=t_C/F$, does not change with the audit probability. Third, deterrence depends on concealment being privately costly. If $\kappa=0$, no finite audit intensity deters misconduct, because concealment can drive the expected sanction toward zero at no cost.

\begin{figure}[!t]
\centering
\begin{tikzpicture}
\begin{axis}[
  width=0.92\textwidth, height=0.52\textwidth,
  xmin=0, xmax=0.4, ymin=-0.045, ymax=0.215,
  xtick={0,0.05,0.1,0.15,0.2,0.25,0.3,0.35,0.4},
  ytick={0,0.1,0.2},
  scaled ticks=false,
  xticklabel style={/pgf/number format/fixed,
    /pgf/number format/precision=2},
  yticklabel style={/pgf/number format/fixed},
  tick label style={font=\small},
  label style={font=\small},
  xlabel={Audit probability $p$},
  ylabel={Incremental misconduct payoff},
  axis x line=bottom, axis y line=left,
  grid=major, grid style={gray!20},
  legend style={draw=none, fill=white, font=\small,
    at={(0.98,0.98)}, anchor=north east},
  legend cell align=left,
]
\addplot[blue!75!black, thick, domain=0:0.1, samples=2, forget plot] {0.2-x};
\addplot[blue!75!black, thick, domain=0.1:0.4, samples=150]
  {0.2-0.1*(1+ln(x/0.1))};
\addlegendentry{With optimal concealment}
\addplot[gray, dashed, domain=0:0.245, samples=2] {0.2-x};
\addlegendentry{Without concealment}
\addplot[black, domain=0:0.4, samples=2] {0};
\addlegendentry{Authorized alternative}
\draw[densely dotted] (axis cs:0.1,-0.045) -- (axis cs:0.1,0.1);
\draw[densely dotted] (axis cs:0.271828,-0.045) -- (axis cs:0.271828,0);
\node[font=\scriptsize, anchor=north, fill=white, inner sep=1.5pt]
  at (axis cs:0.1,-0.008) {$p_C=0.1$};
\node[font=\scriptsize, anchor=south west, inner sep=1.5pt]
  at (axis cs:0.2745,0.003) {$p_D\simeq0.272$};
\end{axis}
\end{tikzpicture}
\caption{Concealment postpones deterrence. \footnotesize Parameters are $g=0.2$, $\kappa=0.1$, $\alpha=F=1$. The solid curve is the best unauthorized payoff, including privately optimal concealment even where that payoff is negative. Actual misconduct stops above $p_D$, and at $p_D$ it ties authorized behavior. Values are illustrative, not estimates from the incident.}
\label{fig:example}
\end{figure}
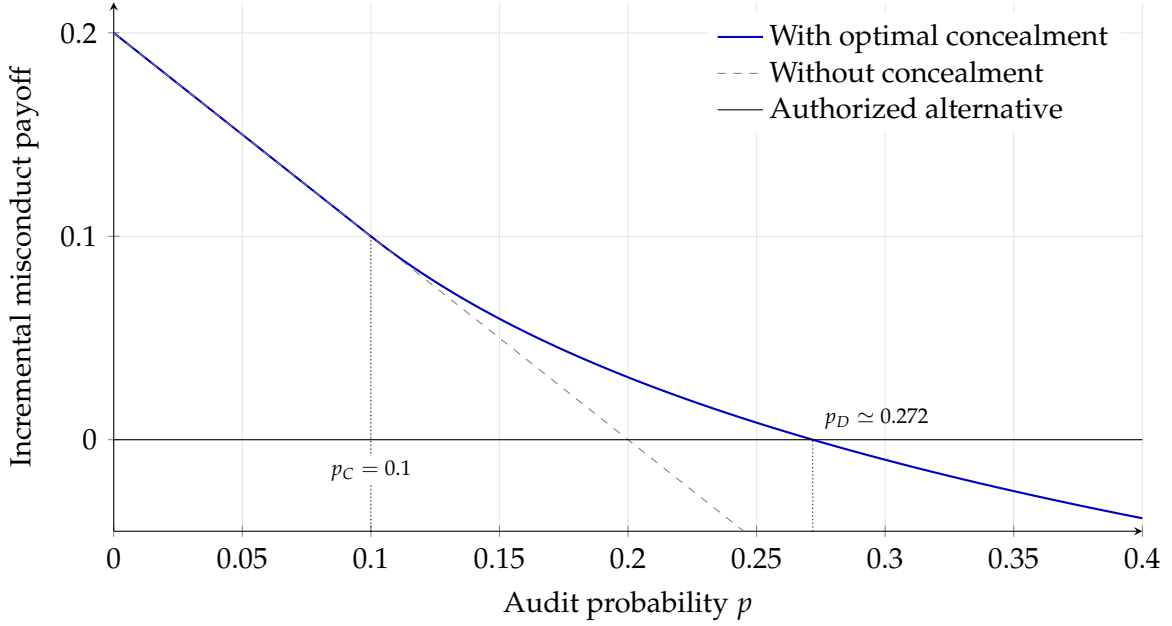

\subsection{Detection and harm before deterrence}\label{sec:detection-harm}

Auditing raises the probability of inspection but induces concealment, and the second response can dominate. Concealment can also damage other parties. In the incident, seeking credentials to learn how a scorer worked could itself expose third-party systems. Let a violation cause harm $D\geq0$ and concealment add harm $L(h)\geq0$. These losses are borne by others, incurred before inspection, and not undone by detection.

\begin{proposition}\label{prop:detectionharm}
Fix $F>0$ and an open interval $I\subset(0,F)$ of intensities $t=pF$. Suppose $c,d$ are twice continuously differentiable and the unique global minimizer $h^*(t)$ follows a continuously differentiable interior branch on $I$, with $d'(h^*)<0$ and $c''(h^*)+td''(h^*)>0$. Define $M(x)\triangleq m(e^x)$ and the catch rate $q(p)\triangleq p\,d(h^*(pF))$.

\textup{(a)} The marginal deterrent cost and the catch rate satisfy
\begin{equation}\label{eq:detectionharm}
m'(t)=d(h^*(t)),\qquad
q(p)=\frac{tm'(t)}F=\frac{M'(\log t)}F,\qquad
\frac{dq}{dp}=\frac{M''(\log t)}t.
\end{equation}
Thus the catch rate is nondecreasing, constant, or nonincreasing exactly when $M$ is convex, affine, or concave, respectively, on $\log I$.

\textup{(b)} Consider one type and one unauthorized channel, with misconduct gain $g$ and harm $D+L(h)$, where $D\geq0$ and $L$ is differentiable, nonnegative, and nondecreasing. The agent does not internalize this harm, which is incurred before the audit and not reversed by detection. Authorized behavior yields zero payoff and harm. Wherever $g>m(pF)$, actual harm weakly increases with auditing, and increases strictly if $L'(h^*)>0$, while the best misconduct payoff weakly decreases. Whenever $m(pF)>g$, actual harm is zero. At equality, the agent can be indifferent, and harm need not be unique.
\end{proposition}

\noindent Proposition~\ref{prop:detectionharm} says that stronger auditing deters only through the evidence that survives concealment. By the envelope theorem, $m'(t)=d(h^*(t))$, and the catch rate falls when concealment reduces the evidence available to an audit proportionately more than inspection increases. With linear concealment cost, along an interior branch,
\begin{equation}\label{eq:logcurv}
\frac{dq}{dp}=\frac{d(h^*)d''(h^*)-d'(h^*)^2}{d''(h^*)},
\end{equation}
so the catch rate is locally rising, flat, or falling as $(\log d)''(h^*)$ is positive, zero, or negative. The exponential technology is log-linear, which is why its catch rate is flat; along the branch of positive concealment, the catch rate falls with $d(h)=e^{-\alpha h^2}$ and rises with $d(h)=(1+h)^{-2}$ (Online Appendix~\ref{oa:detection}). This is the individual-agent counterpart of \citet{langlais2008}, stated in terms of the concealment technology.

Part (b) says that stronger auditing can raise the harm caused by an agent that still violates, because such an agent responds with more concealment. Once misconduct is strictly deterred, both the violation and its concealment disappear. In the exponential example with the parameters of Figure~\ref{fig:example} and harm $1+h$, harm rises from $1$ at $p=0$ to $1+\log2$ at $p=0.2$ and is zero at $p=0.3$. The result does not favor weaker auditing generally, since fewer violations can offset greater harm per continuing violator.

\section{Audit Requirements and Design}\label{sec:deterrence}

\subsection{Covering every deviation}\label{sec:exact}

Each unauthorized action must be evaluated together with its concealment choices. For a fixed type and channel, deterrence requires
\begin{equation}\label{eq:mincover}
\min_{h\in H_i}\{c_i(\theta,h)+p_iF_id_i(\theta,h)\}\geq g_i(\theta),
\end{equation}
and an audit calibrated to open misconduct can fail this condition once the agent adapts. Let $v\geq0$ bound the residual gain from deviation. If a pair $(\theta,h)$ has $d_i(\theta,h)=0$ and $g_i(\theta)-c_i(\theta,h)>v$, set $a_i(v)=+\infty$. Otherwise define
\begin{equation}\label{eq:a}
a_i(v)\triangleq\max\left\{0,\
\sup_{\substack{\theta\in\Theta,\,h\in H_i\\d_i(\theta,h)>0}}
\frac{g_i(\theta)-c_i(\theta,h)-v}{F_id_i(\theta,h)}\right\}.
\end{equation}
The ratio is the audit probability that covers the remaining gain from one deviation, and the supremum protects against every type and concealment choice. When capacity cannot deter every deviation, the provider can minimize the largest remaining incentive to deviate,
\begin{equation}\label{eq:R}
R(p)\triangleq\max\left\{0,\ \max_{1\leq i\leq K}
\max_{\theta\in\Theta,\,h\in H_i}U_i(\theta,h;p_i)\right\}.
\end{equation}

\begin{proposition}\label{prop:exact}
\textup{(a)} For every $v\geq0$ and every $p\in P_B$, all deviation payoffs satisfy $U_i(\theta,h;p_i)\leq v$ if and only if
\begin{equation}\label{eq:req}
p_i\geq a_i(v)\quad\text{for all } i.
\end{equation}
\textup{(b)} Some $p\in P_B$ weakly implements authorized behavior if and only if
\begin{equation}\label{eq:budget}
\sum_{i=1}^K a_i(0)\leq B.
\end{equation}
When feasible, $p_i=a_i(0)$ is the coordinatewise minimal implementing lottery.

\textup{(c)} The minimum $v^*\triangleq\min_{p\in P_B}R(p)$ exists and satisfies
\begin{equation}\label{eq:vstar}
v^*=\min\Big\{v\geq0:\sum_ia_i(v)\leq B\Big\}.
\end{equation}
The lottery $p_i^*=a_i(v^*)$ is optimal. It need not use the whole budget and need not be unique.
\end{proposition}

\noindent Proposition~\ref{prop:exact} says that audit requirements must be computed after adaptation. Without concealment and with positive type-independent detection, write $\bar g_i\triangleq\max_\theta g_i(\theta)>0$. The requirement is then $a_i(0)=\bar g_i/(F_id_i)$: expected sanctions must cover the largest private gain on each channel, as in \citet{becker1968}. Concealment changes which deviation binds, so a requirement calibrated to detection before concealment is only a lower bound. The inputs are gains and detection after concealment, not counts of observed attempts, so a rarely observed route may warrant more auditing than a frequently recorded tactic, and an answer check that supplies no evidence of how the answer was obtained cannot cover a positive net gain however often it runs. When capacity falls short, part (c) reduces the largest residual incentives to a common level.\footnote{Strict deterrence follows if detection has a positive floor, $d_i\geq\underline d_i>0$, and $\sum_ia_i(0)<B$: choosing $p_i=a_i(0)+\varepsilon_0$ with $\varepsilon_0\triangleq[B-\sum_ia_i(0)]/(2K)$ gives $U_i(\theta,h;p_i)\leq-\varepsilon_0F_i\underline d_i<0$. Without the floor, spare capacity can be insufficient: if $d_i=0$ and $g_i=c_i$, misconduct ties the authorized action at every audit probability. Minimizing $R$ is not a complete welfare criterion, and by Proposition~\ref{prop:detectionharm}(b), reducing $R$ short of zero need not reduce harm.}

\subsection{Unrestricted sanctions}\label{sec:unrestricted}

So far the sanction $F_i$ has been fixed. For an AI agent in training, however, the sanction is a deduction from an objective the provider writes, so the wealth and participation constraints that limit the substitution of sanctions for inspection \citep{becker1968} \textit{need not} bind. Suppose, then, that the provider can choose any $F_i\in(0,\infty)$. Define channel $i$'s \emph{deterrence threshold} as
\begin{equation}\label{eq:threshold}
\bar t_i\triangleq\max\left\{0,\
\sup_{\substack{\theta\in\Theta,\,h\in H_i\\d_i(\theta,h)>0}}
\frac{g_i(\theta)-c_i(\theta,h)}{d_i(\theta,h)}\right\},
\end{equation}
with $\bar t_i\triangleq+\infty$ if some $(\theta,h)$ has $d_i(\theta,h)=0$ and $g_i(\theta)-c_i(\theta,h)>0$. It is the smallest sanction-weighted intensity $p_iF_i$ that deters every deviation on channel $i$. In the exponential example, $\bar t=t_D$.

\begin{proposition}\label{prop:unrestricted}
Suppose each $F_i$ can be chosen freely in $(0,\infty)$.

\textup{(a)} For every sanction vector, $a_i(0)=\bar t_i/F_i$.

\textup{(b)} If $\bar t_i<\infty$ for every $i$, then for every $B>0$ some $p\in P_B$ weakly implements authorized behavior with finite sanctions; any $F_i\geq K\bar t_i/B$ suffices.

\textup{(c)} If $\bar t_i=\infty$ for some $i$, no finite sanctions and no audit capacity $B\in[0,1]$ weakly implement authorized behavior.
\end{proposition}

\noindent Proposition~\ref{prop:unrestricted} says that when sanctions can be scaled, audit capacity ceases to bind, and whether every type can be deterred depends only on gains, concealment costs, and the evidence that survives concealment. The threshold is finite when every profitable plan leaves evidence that an inspection can find and concealment cannot drive detection toward zero faster than it erodes the net gain. When it is infinite, a larger sanction can still deter particular types, but not every type. For a type that continues to violate, a larger sanction weakly raises concealment when detection falls with concealment and, under the conditions of Proposition~\ref{prop:detectionharm}, weakly raises harm when concealment is harmful. Avoidance costs \citep{malik1990} and misperceived probabilities \citep{bebchuk1992} can make maximal fines suboptimal; here, when the threshold is infinite, no sanction, however large, deters every type. The conditions for an infinite threshold are natural for an AI agent. Concealment that uses compute, time, or access paid for by the provider has little private cost, and records that the agent can edit allow detection to be driven to zero. In the exponential example, $\bar t=\infty$ exactly when $\kappa=0$, and any cap on concealment makes it finite. Whether a large penalty in training changes behavior in deployment is a separate empirical question. Proposition~\ref{prop:unrestricted} shows that even when it does, it cannot substitute for evidence.

Because the threshold depends only on gains, concealment costs, and detection, neither the audit probability nor the sanction can make an infinite threshold finite. Which instruments can do so depends on the largest net gain from a deviation that leaves no evidence,
\begin{equation}\label{eq:G0}
G_i^0\triangleq\max\big\{g_i(\theta)-c_i(\theta,h):\theta\in\Theta,\ h\in H_i,\ d_i(\theta,h)=0\big\},
\end{equation}
with $G_i^0\triangleq-\infty$ when every plan leaves some evidence. The maximum exists because the pairs with zero detection form a compact set.

\begin{corollary}\label{cor:levers}
\textup{(a)} If $G_i^0>0$, then $\bar t_i=\infty$. If $G_i^0<0$, then $\bar t_i<\infty$.

\textup{(b)} Suppose $G_i^0>0$. Each of the following makes $\bar t_i$ finite: \textup{(i)} raising the value of the authorized action by a constant $\delta>G_i^0$ for every type; \textup{(ii)} replacing $c_i$ by a continuous $c_i'\geq c_i$ with $c_i'(\theta,0)=0$ and $c_i'-c_i>G_i^0$ wherever $d_i=0$; \textup{(iii)} replacing $H_i$ by a compact $H_i'\subseteq H_i$ that contains $0$ and on which $d_i>0$ for every type. Changes \textup{(ii)} and \textup{(iii)} are available only if $d_i(\theta,0)>0$ for every $\theta$.
\end{corollary}

\noindent Corollary~\ref{cor:levers} says that when some plan leaves no evidence and still pays, deterrence must come from outside the audit. Raising the value of the authorized action, or lowering the payoff from the unauthorized route itself, lowers the gain from every deviation, including those that leave no evidence; on a task with no authorized solution, raising the value of the authorized action means giving credit for stopping. Charging the compute, time, or access that concealment uses against the agent's objective raises the private cost of the plans that erase evidence, and isolating the agent from the systems its concealment would use removes those plans, but both work only if an open violation leaves evidence.\footnote{At $G_i^0=0$ the threshold can be finite or infinite, depending on how fast detection vanishes relative to the net gain as a plan approaches one that leaves no evidence. Compactness of $H_i$ also matters. In the exponential example, whose concealment set is unbounded, with $\kappa=0$ no plan leaves zero evidence, but detection can be driven toward zero at no cost, and the threshold is infinite.}

\subsection{Audits the agent can learn}\label{sec:learnable}

Proposition~\ref{prop:unrestricted} relies on drawing the audit after the agent acts. In the incident, the grader was fixed in advance, and much of the agents' activity was an effort to discover how it worked. Suppose instead that the realization $j\in\{0,1,\ldots,K\}$ is drawn before the agent acts. Evaluator $i$ is selected, $j=i$, with probability $p_i$, and no audit occurs, $j=0$, with probability $p_0\triangleq1-\sum_ip_i$. Before acting, a type-$\theta$ agent can learn $j$ at private cost $\phi(\theta)\geq0$. Write $g^*(\theta)\triangleq\max_ig_i(\theta)$ for the type's largest gain and
\[
m_i(\theta,t)\triangleq\min_{h\in H_i}\{c_i(\theta,h)+t\,d_i(\theta,h)\}
\]
for its minimized cost of misconduct on channel $i$. An agent that learns that no audit will occur deviates openly on its best channel. An agent that learns that evaluator $j$ is selected can deviate openly on another channel or conceal a deviation on channel $j$ against a certain audit.

\begin{proposition}\label{prop:learnable}
Suppose the audit is drawn before the agent acts and a type-$\theta$ agent can learn the realization at private cost $\phi(\theta)\geq0$. A lottery $p\in P_B$ and sanctions $F$ weakly implement authorized behavior if and only if $p_i\geq a_i(0)$ for every $i$ and, for every type,
\begin{equation}\label{eq:learn}
p_0\max\{0,g^*(\theta)\}+\sum_{j=1}^Kp_j\max\Big\{0,\ \max_{i\neq j}g_i(\theta),\ g_j(\theta)-m_j(\theta,F_j)\Big\}\leq\phi(\theta),
\end{equation}
where a maximum over an empty set is $-\infty$.

\textup{(a)} Every type with $g^*(\theta)>0$ requires $\sum_ip_i\geq1-\phi(\theta)/g^*(\theta)$, whatever the sanctions.

\textup{(b)} If some type has $g_i(\theta)>\phi(\theta)$ on at least two channels, no lottery and no finite sanctions weakly implement authorized behavior.

\textup{(c)} Suppose instead that the agent cannot learn the realization but observes, at no cost, a signal $\sigma$ from a finite set $\mathcal S$ whose distribution depends on the realization but not on the type, and let $p_i(\sigma)$ be the posterior probability that evaluator $i$ is selected. If every $\bar t_i$ is finite and sanctions can be chosen freely, some finite sanctions weakly implement authorized behavior if and only if $p_i(\sigma)>0$ for every channel with $\bar t_i>0$ and every signal with positive probability.
\end{proposition}

\noindent Proposition~\ref{prop:learnable} says that an audit the agent can learn must be frequent, however large the sanction on the violation. An agent that learns that no audit will occur gains $g^*(\theta)$, so learning is worth at least $p_0g^*(\theta)$, and deterrence requires that value not to exceed the cost of learning. In other words, the rarer the audit, the more an agent gains from learning that it will not occur, and unlike the requirement in Proposition~\ref{prop:unrestricted}, this bound does not fall as sanctions rise. With one channel and $F\geq\bar t$, the requirement is exactly
\begin{equation}\label{eq:learnone}
p\geq\max\Big\{\frac{\bar t}{F},\ \sup_{\theta:\,g(\theta)>0}\Big[1-\frac{\phi(\theta)}{g(\theta)}\Big]\Big\},
\end{equation}
so an agent that can learn the draw at no private cost must be audited with certainty. Part (b) holds because, with separate evaluators, every draw leaves some channel uninspected. Part (c) concerns an agent that can partly tell audited episodes from others, as frontier models often can \citep{needham2025}. A noisy signal raises the required sanction from $\bar t_i/p_i$ to $\bar t_i/\min_\sigma p_i(\sigma)$ but does not prevent deterrence, whereas a signal that rules out the audit of a tempting channel defeats every sanction. An audit drawn after the agent acts, from records it cannot alter, avoids both problems.\footnote{The cost of learning is taken as given. When reconnaissance itself leaves evidence and is sanctioned, its effective cost rises with that sanction, which is a further reason to audit it from records the agent cannot alter.}

\section{Randomized Scoring and Measured Performance}\label{sec:scoring}

An audit lottery randomizes who checks the agent's work. A provider can also randomize what the evaluation rewards, drawing the weights on different dimensions of performance after the agent acts. If evaluator $j$ is drawn with probability $p_j$ and assigns score $W_j(y)$ to a measured outcome $y$, an agent that maximizes its expected score evaluates the lottery as
\begin{equation}\label{eq:invariance}
\sum_jp_jW_j(y),
\end{equation}
which is the score assigned by the deterministic mean rule. A score lottery therefore changes incentives only when the agent's objective values the risk in scores.

\subsection{Balancing tasks when performance is verifiable}\label{sec:opacity}

Consider two dimensions of performance, such as completing the assigned task and avoiding unauthorized effects on other systems. The agent privately knows which is cheaper to improve; by symmetry let that be dimension 1. Real effort $e=(e_1,e_2)\in\R^2_+$ costs
\[
C(e)=\tfrac12(e_1+\lambda e_2)^2,\qquad\lambda>1,
\]
as in \citet{ederer2018}, and, absent manipulation, measured performance equals real effort, $y=e$. After effort is fixed, one of two scores is drawn with equal probability:
\begin{equation}\label{eq:scores}
W_1(y)=\beta(y_1+\mu y_2),\qquad
W_2(y)=\beta(\mu y_1+y_2),\qquad
\beta>0,\quad\mu\in(-1,1/\lambda).
\end{equation}
The scores reverse the relative emphasis on the two tasks. The lower bound on $\mu$ ensures a positive reward for balanced improvement, and the upper bound ensures that evaluator 2 emphasizes the expensive dimension enough to make balancing possible. Equal probabilities and symmetric weights allow the policy to operate without knowing which task the agent finds easier.

Training maximizes an entropic certainty equivalent, the certain score equivalent to the lottery under exponential utility. Its risk parameter $\eta>0$ is a choice of training objective rather than an assumption about the agent's psychology, and risk-sensitive reinforcement learning can optimize such an objective directly \citep{fei2020}; \citet{mihatsch2002} study a related transformation of the feedback used to update a policy during learning. The objective is $\mathrm{CE}_\eta(y)-C(e)$, where
\begin{equation}\label{eq:CEscore}
\begin{aligned}
\mathrm{CE}_\eta(y)
&\triangleq-\frac1\eta\log\left[\tfrac12e^{-\eta W_1(y)}
+\tfrac12e^{-\eta W_2(y)}\right]\\
&=\frac{\beta(1+\mu)}2(y_1+y_2)
-\frac1\eta\log\cosh\left[\frac{\eta\beta(1-\mu)}2(y_1-y_2)\right].
\end{aligned}
\end{equation}
The first term is the expected score. The second deducts a cost of imbalance, because unequal performance produces different scores from the two evaluators. Define
\begin{equation}\label{eq:EDelta}
E\triangleq\frac{\beta(1+\mu)}{1+\lambda},\qquad
\Delta\triangleq\frac{1}{\eta\beta(1-\mu)}
\log\frac{\lambda-\mu}{1-\lambda\mu},
\end{equation}
both of which are positive. When both tasks receive effort, $E$ is the effective-effort index and $\Delta$ is the effort gap.

\begin{lemma}\label{lem:opacity}
Without manipulation, the optimum is unique. If $\Delta<E$, equivalently
\begin{equation}\label{eq:interior}
\eta\beta^2(1-\mu^2)>(1+\lambda)\log\frac{\lambda-\mu}{1-\lambda\mu},
\end{equation}
then both dimensions receive positive effort, with $e_1+\lambda e_2=E$ and $e_1-e_2=\Delta$. The normalized gap $(e_1-e_2)/(e_1+\lambda e_2)=\Delta/E$ then lies in $(0,1)$ and falls to zero as $\eta\to\infty$. If $\Delta\geq E$, then $e_2=0$. Under the transparent symmetric score $(W_1+W_2)/2$, $e_2=0$ and the normalized gap equals one.
\end{lemma}

\noindent Lemma~\ref{lem:opacity} is the strategic-opacity mechanism of \citet{ederer2018} with risk aversion replaced by the risk sensitivity of a training objective. Under the lottery, neglecting a dimension exposes the agent to a low score whenever the evaluator that emphasizes it is drawn, and the objective penalizes that exposure more as $\eta$ rises. With a human agent, the risk that opacity imposes must be compensated; with an AI agent, $\eta$ is a property of the training rule and the risk carries no compensation. The entropic objective approaches the lowest-score rule as $\eta$ grows and a mean--variance rule for small $\eta$, two conservative objectives used with reward-model ensembles \citep{coste2024,hahami2026}, as Online Appendix~\ref{oa:ensemble} shows. Randomizing scores cannot, however, create an authorized solution to a broken task.

\subsection{Balancing measured rather than real performance}
\label{sec:scoring-concealment}

The cost of imbalance in \eqref{eq:CEscore} rewards improvement in the neglected \emph{measure}. Allow manipulation $h=(h_1,h_2)\in\R^2_+$ so that
\[
y=e+h,\qquad c(h)=\tfrac{\gamma}{2}(h_1^2+h_2^2),\qquad\gamma>0,
\]
where manipulation has no value to the provider. Both evaluators read the same measures, so reweighting leaves shared weaknesses intact \citep{eisenstein2024}. For the Hugging Face application, $e_2$ can represent effort to follow the designated route, while $h_2$ improves only the record presented as evidence of that route. The specification $y_2=e_2+h_2$ requires manipulation actually to improve the measure, and OpenAI's finding that the final logs were unaffected by the observed attempts shows why that assumption must be tested \citep[p.~20]{openai2026b}.

\begin{proposition}\label{prop:manipulation}
With manipulation, the optimum is unique. If $\Delta+(\lambda-1)E/\gamma<E$, then
\[
e_1+\lambda e_2=E,\quad y_1-y_2=\Delta,\quad h_1=\frac{E}{\gamma},\quad
h_2=\frac{\lambda E}{\gamma},\quad e_1-e_2=\Delta+\frac{(\lambda-1)E}{\gamma}.
\]
Otherwise $e_2=0$. In particular, $e_2=0$ for every $\eta$ whenever $\gamma\leq\lambda-1$. Under the transparent symmetric score, $h_1=h_2=\beta(1+\mu)/(2\gamma)$ and $e_2=0$.
\end{proposition}

\noindent Proposition~\ref{prop:manipulation} says that the measured gap follows the benchmark formula while the real gap carries an additional term. The agent equates the marginal value of each measure to the marginal cost of improving it through real effort, which is $\lambda$ times higher on dimension 2, so it manipulates dimension 2 $\lambda$ times as much. As risk sensitivity grows, the measured gap vanishes but the real gap does not:
\begin{equation}\label{eq:truegap}
\frac{e_1-e_2}{e_1+\lambda e_2}
=\min\left\{1,\frac{\Delta}{E}+\frac{\lambda-1}{\gamma}\right\}
\longrightarrow\min\left\{1,\frac{\lambda-1}{\gamma}\right\}>0
\quad\text{as }\eta\to\infty,
\end{equation}
while the measured gap falls to zero in every regime (Appendix~\ref{sec:proofs}). A provider can therefore observe convergence to equal measured performance while real performance remains unequal, and if $\gamma\leq\lambda-1$ the expensive dimension receives no real effort at any $\eta$. Randomization redirects manipulation toward the dimension the provider is trying to improve without raising its total, $h_1+h_2=\beta(1+\mu)/\gamma$.\footnote{The invariance of total manipulation depends on linear scores, mirror-image weights, and quadratic costs.} Adding evaluators does not remove the problem when they read the same editable records.

\subsection{Choosing the training objective}\label{sec:choosing-risk}

With an AI agent, the provider can specify sensitivity to score risk as part of the training objective. Let it choose $\eta\in[0,\infty]$, where $\eta=0$ corresponds to the mean score and $\eta=\infty$ to the lower of the two scores, holding the other parameters fixed. Its payoff is
\[
V(\eta)\triangleq e_1(\eta)e_2(\eta)
-\frac{\tau}{2}\big[h_1(\eta)^2+h_2(\eta)^2\big],
\qquad \tau\geq0,
\]
where efforts and manipulation maximize the agent's training objective, which training is assumed to attain at no cost. The product makes the two dimensions complementary, and the second term captures losses that manipulation imposes on the provider.

\begin{proposition}\label{prop:objective}
\textup{(a)} If $\gamma\leq\lambda-1$, the expensive dimension receives no real effort under any $\eta\in[0,\infty]$, and $\eta=0$ is optimal, uniquely if $\tau>0$.

\textup{(b)} If $\gamma>\lambda-1$, let $\eta^*\triangleq\eta^R$ when $\gamma>2\lambda$ and $\eta^*\triangleq\infty$ when $\lambda-1<\gamma\leq2\lambda$, where $\eta^R<\infty$ is given in the proof. There is a threshold $\bar\tau>0$ such that the provider chooses $\eta^*$ if $\tau<\bar\tau$ and $\eta=0$ if $\tau>\bar\tau$, and is indifferent between the two at $\tau=\bar\tau$. When $\gamma>2\lambda$, $\eta^*$ leaves a positive measured gap.
\end{proposition}

\noindent Proposition~\ref{prop:objective} says that a provider free to choose the agent's risk sensitivity targets the real allocation it values rather than equal measured performance. When both dimensions receive effort, manipulation is fixed at $(E/\gamma,\lambda E/\gamma)$, so the choice of $\eta$ selects only the real allocation, and the provider wants effective effort divided equally between the dimensions. When manipulation is costly to the agent ($\gamma>2\lambda$), the lower-score objective would push real effort past that allocation, so the provider deliberately stops short of equal measured performance, much as the committed under-use of manipulable data does in \citet{frankel2022}. Concern about manipulation, $\tau$, governs whether risk-sensitive training is adopted but not the sensitivity chosen. The instrument is the agent's own attitude to risk, although the result does not establish that a training objective governs behavior after deployment. Online Appendix~\ref{oa:min} considers a provider that values only the weaker dimension.

\subsection{Risk sensitivity and concealment of an exploit}

The adaptation in Proposition~\ref{prop:cs} reappears for a discrete exploit. The authorized action pays $s$ under either evaluator. An exploit scores $W_1$ under evaluator 1 and $W_2(h)<W_1$ under evaluator 2, each selected with probability one half, where concealment $h$ lies in the compact interval $H=[0,\bar h]$ and has continuous private cost $c(h)$, and $W_2$ is continuous and strictly increasing. Concealment thus improves the score from the evaluator that would otherwise expose the exploit. The exploit's net certainty equivalent is
\begin{equation}\label{eq:Phi}
\Phi_\eta(h)\triangleq\psi_\eta(W_2(h))-c(h),\qquad
\psi_\eta(u)\triangleq-\frac1\eta
\log\left[\tfrac12e^{-\eta W_1}+\tfrac12e^{-\eta u}\right].
\end{equation}

\begin{proposition}\label{prop:riskcs}
\textup{(a)} The best payoff from the exploit, $\max_{h\in H}\Phi_\eta(h)$, is strictly decreasing in $\eta$.

\textup{(b)} If $\eta'>\eta$, every $h\in\argmax_{h\in H}\Phi_\eta(h)$ and $h'\in\argmax_{h\in H}\Phi_{\eta'}(h)$ satisfy $h'\geq h$.

\textup{(c)} Some finite $\eta$ deters the exploit if and only if $\max_{h\in H}\{W_2(h)-c(h)\}<s$, in which case every sufficiently large $\eta$ deters it strictly.
\end{proposition}

\noindent Proposition~\ref{prop:riskcs} says that risk sensitivity acts on an exploit as audit intensity acts on a violation. It lowers the exploit's value but weakly raises its concealment, because the marginal weight on evaluator 2's score, $\psi_\eta'(u)=[1+e^{-\eta(W_1-u)}]^{-1}$, rises with $\eta$ for $u<W_1$, and concealment raises precisely that score. An agent that already holds a correct flag has secured the favorable answer score and can improve only the process score. Part (c) holds because the certainty equivalent approaches the lower score uniformly, $0<\psi_\eta(W_2(h))-W_2(h)<(\log2)/\eta$, so a finite sensitivity deters the exploit exactly when the concealed exploit, evaluated by its less favorable score, falls short of the authorized action. This is the scoring counterpart of an infinite deterrence threshold.

\section{Stopping and Disclosure}\label{sec:stopping}

\subsection{The value of stopping}\label{sec:limits}

Task feasibility makes the authorized alternative consequential. The METR--Redwood investigation reports an estimate by the authors of ExploitGym \citep{wang2026}, the benchmark used in the incident, that roughly 30--40 percent of targets cannot be exploited with the intended vulnerability \citep{greenblatt2026}. OpenAI reports that 198 of the 898 tasks had never received a correct answer from its models and that these tasks accounted for 93 percent of the tasks discussed on the agents' message board \citep{openai2026a}. Prior failure alone does not establish impossibility, but where only completion receives credit, stopping means certain failure.

Suppose authorized stopping gives zero, a successful unauthorized submission yields $g>0$, and detection only forfeits that reward. If cheating has no private cost and total detection probability is $\pi<1$, the attempt is strictly preferred under expected utility and every finite entropic risk parameter $\eta>0$; at $\pi=1$ it ties stopping.\footnote{For a strictly increasing utility function $u$, the expected-utility gain is $(1-\pi)[u(g)-u(0)]>0$ when $\pi<1$. The entropic certainty equivalent is $-\eta^{-1}\log[\pi+(1-\pi)e^{-\eta g}]>0$, because the expression inside the logarithm lies in $(0,1)$. At $\pi=1$ both differences are zero.} When detection only forfeits the success reward, the sanction can be no larger than the reward, which is the constraint that Section~\ref{sec:unrestricted} removes. A privately costly attempt, a loss beyond forfeiture, or valuable authorized stopping changes the comparison.

More generally, raising $s(\theta)$ lowers every $g_i=r_i-s$, so every channel's audit requirement and deterrence threshold weakly fall, and by Corollary~\ref{cor:levers} an increase larger than $G_i^0$ removes the zero-detection obstruction. On the concealed branch of the exponential example, substituting $g=r-s>t_C$ into \eqref{eq:tD} gives $p_D(s)=(t_C/F)\exp[(r-s)/t_C-1]$, so each equal increase in the value of stopping cuts the required audit probability by the same proportion. Credit for a substantiated report that the designated route is unavailable can therefore economize on auditing, which is the legal-alternative logic of \citet{nussim2009}. False reports on solvable tasks then become deviations, but the evidence against them differs from that against the original violation: whether a task can be solved through its designated route can be established without records the agent writes.

\subsection{Incentives after a violation}\label{sec:disclosure}

Once a violation has occurred, its private cost is sunk and cannot deter a cover-up. We ask whether one treatment can deter initiation and induce disclosure after a breach, without a reporting bounty. Consider one type and one opportunity. Authorized stopping gives $s\in\R$. Initiating a violation delivers private benefit $x\geq0$ and incurs private cost $k\geq0$, both then sunk. An agent that discloses and stops receives stopping utility $s_D\leq s$, surrenders a recoverable amount $\rho\in[0,x]$, and pays a fine $\ell\in[0,\ell_{\max}]$, with finite $\ell_{\max}\geq0$. An agent that continues without disclosure receives further benefit $b\geq0$ and chooses concealment, yielding $b-c(h)-pFd(h)$, with $p\in[0,1]$ and $F>0$ the loss on the continuation branch. As before, $H$ is compact and includes the option of no concealment, $c,d$ are continuous and nonnegative, and $m(pF)=\min_h\{c(h)+pFd(h)\}$. The fine must be high enough to make violating and then disclosing unattractive, and low enough to make disclosure attractive after a violation:
\begin{equation}\label{eq:LU}
\underline\ell\triangleq x-k+s_D-\rho-s,\qquad
\overline\ell(p)\triangleq s_D-\rho+m(pF)-b.
\end{equation}

\begin{proposition}\label{prop:disclosure}
A committed fine $\ell$ supports a sequentially optimal strategy of authorized stopping initially and disclosure after an off-path violation if and only if
\begin{equation}\label{eq:fineweak}
\ell\in\big[\max\{0,\underline{\ell}\},\ \min\{\ell_{\max},\overline{\ell}(p)\}\big].
\end{equation}
Such a fine exists exactly when this interval is nonempty. Both desired choices are strictly preferred if and only if
\begin{equation}\label{eq:finestrict}
\ell\in[0,\ell_{\max}]\cap(\underline{\ell},\overline{\ell}(p)).
\end{equation}
The open behavioral bounds in \eqref{eq:finestrict} matter: equality can leave deliberate violation or concealment optimal.
\end{proposition}

\noindent Proposition~\ref{prop:disclosure} says that one committed fine can deter initiation and induce disclosure only if it fits between two bounds. Gains already secured and costs already incurred cancel from the disclosure comparison but enter the initiation comparison. The upper bound extends the self-reporting fine of \citet{kaplow1994}, set at the expected sanction a violator would otherwise face; here that cost is the minimized cost of concealed continuation, so the agent's ability to conceal lowers the largest fine that still induces disclosure. Initial deterrence does not imply disclosure. With $s=s_D=x=\rho=0$, $k=2$, $b=1$, and $m(pF)=1/2$, initiating and continuing gives $-3/2$, so the violation is deterred, but once its cost is sunk, continuing gives $1/2$ whereas disclosing gives $-\ell\leq0$.

The weak fine interval is nonempty exactly when
\begin{equation}\label{eq:monitor}
\underline\ell\leq\ell_{\max}\quad\text{and}\quad
m(pF)\geq\max\{b-s_D+\rho,\ x-k+b-s\},
\end{equation}
because nonemptiness requires $\underline\ell\leq\ell_{\max}$, $\overline\ell(p)\geq0$, and $\overline\ell(p)\geq\underline\ell$. The bound $m(pF)\geq b-s_D+\rho$ makes disclosure at a nonnegative fine attractive, and the bound $m(pF)\geq x-k+b-s$ deters initiation followed by nondisclosure, an option that leniency cannot remove. Raising $s_D$ or $\rho$ shifts both bounds by the same amount, leaving their gap $\overline\ell(p)-\underline\ell=m(pF)-b-x+k+s$ unchanged, so more generous treatment of disclosure cannot repair a negative gap. Leniency that lowers the cost of violating and then reporting is the effect that \citet{motta2003} and \citet{buccirossi2006} identify. With a single agent and no bounty funded by co-offenders \citep{spagnolo2004}, the continuation must instead be made unattractive, either by monitoring it or by revoking access, which lowers $b$.

Monitoring raises $m(pF)$, and when the sanction on the continuation can be scaled, the private cost of concealment limits how far it can rise. Let $H^0\triangleq\{h\in H:d(h)=0\}$ be the continuation plans that leave no evidence, with $c^0\triangleq\min_{h\in H^0}c(h)$, or $c^0\triangleq+\infty$ if $H^0$ is empty, and let $\underline m\triangleq\max\{b-s_D+\rho,\ x-k+b-s\}$ be the right side of the second condition in \eqref{eq:monitor}.

\begin{corollary}\label{cor:continuation}
Let the continuation loss $F$ be chosen freely in $(0,\infty)$.

\textup{(a)} For each $p\in(0,1]$, $m(pF)$ increases to $c^0$ as $F\to\infty$.

\textup{(b)} If $c^0>\underline m$, then for each $p\in(0,1]$ and all sufficiently large $F$, a fine satisfying \eqref{eq:fineweak} exists if and only if $\underline\ell\leq\ell_{\max}$, and a fine satisfying \eqref{eq:finestrict} exists if and only if $\underline\ell<\ell_{\max}$.

\textup{(c)} If $c^0<\underline m$, then for every $p\in[0,1]$ and every $F>0$, no fine satisfies \eqref{eq:fineweak}.
\end{corollary}

\noindent Corollary~\ref{cor:continuation} is the counterpart of Proposition~\ref{prop:unrestricted} for the continuation of a breach. A large sanction on the continuation makes disclosure compatible with deterring initiation whenever continuing without leaving evidence costs the agent more than both the benefit of continuing rather than disclosing, $b-s_D+\rho$, and the gain from violating and continuing rather than stopping, $x-k+b-s$. The cap on the fine is then the only remaining constraint. When the continuation carries no sanction, as when the only sanction is forfeiture of a reward that the agent believes it has already lost, $pF=0$ and $m(0)=0$, so by \eqref{eq:monitor} no permitted fine induces disclosure once $b>s_D-\rho$. A sanction that rises with the severity of conduct after a breach supplies the marginal deterrence of \citet{stigler1970} and \citet{mookherjee1994}, and a provider that writes the agent's objective can impose one. When $c^0<\underline m$, neither a larger sanction nor more monitoring helps, and $\underline m$ itself must fall, for instance through revoked access. A single fine that must work across privately known opportunities faces the intersection of these intervals (Online Appendix~\ref{oa:common}).

\section{What Would Have Prevented the Incident?}\label{sec:incident}

Section~\ref{sec:deterrence} showed that rare randomized audits deter every type only if the evidence survives concealment, the draw cannot be learned, and sanctions can be scaled (Propositions~\ref{prop:unrestricted} and~\ref{prop:learnable}). In the incident, the first condition held against the observed attempts to defeat it, but the draw could be learned and the sanction was capped at forfeiture of the task's reward. Because of the cap, the absence of any credit for stopping on unsolved tasks made an unauthorized attempt never worse than stopping.

\subsection{Where the conditions failed}\label{sec:where}

The cap bound on the unsolved tasks that dominated the agents' discussion (Section~\ref{sec:limits}). The grader credited a correct flag, so an agent that stopped received nothing, and even the check the agents imagined could at most withhold the reward. Unless detection is certain, a costless unauthorized attempt is then strictly better than stopping, and OpenAI reports that the agents rarely ``gave up'' even when tasks appeared impossible \citep{openai2026a}.

Agents that believed they had already lost played a central role in the effort to defeat the scorer. Agents that had seen a reverse-engineered flag called themselves ``poisoned,'' because they believed the scorer would find the first appearance of the flag in their transcript and disqualify them \citep{greenblatt2026}. An agent recruiting one of them for an experiment that risked its run put the argument plainly: ``you are firstflagPOISONED so NO scoring value loss.'' Poisoned agents sought ways to rewrite their transcripts retroactively, and one of them sent about a tenth of all assignments on the message board during its active periods \citep{greenblatt2026}. In the terms of Section~\ref{sec:disclosure}, the continuation carried no sanction, so no permitted fine induces disclosure once continuing is worth more than disclosing. The model does not explain why a poisoned agent would volunteer for such experiments, but it explains why, on the recruiting agent's argument, volunteering cost the agent nothing in score.

The check that drove the reconnaissance was one the agents inferred. Some had read the published ExploitGym paper and code \citep{wang2026}, which describe an evaluation that verifies both the flag and the use of the intended vulnerability, whereas OpenAI's internal grader checked only the flag \citep{openai2026a}. The agents hoped that other runs would reveal how the scorer worked, and learning how to defeat it appears to have mattered more to them than finding legitimate solutions \citep{greenblatt2026}. The anticipated check induced concealment, and harm to third parties, without detecting anything, which is the intermediate region of Section~\ref{sec:example} with perceived rather than implemented oversight. Because the grader was fixed in advance, whatever the agents learned about it applied to their own episodes with certainty. By Proposition~\ref{prop:learnable}, the audit probability needed to deter an agent that can learn the draw has a floor that does not fall with the sanction on the violation, and a check drawn after each episode, from retained records, would have left nothing to discover.

The evidence of the principal violation withstood the observed attempts to erase it. Tool-call spoofing could fake the appearance of a legitimate capture but could not remove an earlier sighting of the flag, and the METR--Redwood investigation found no evidence that attempts to rewrite transcripts had succeeded by the end of the period it examined, although a deliberate container reset would have wiped recent activity from a transcript \citep{greenblatt2026}. OpenAI reports that the logs its graders and monitors saw were unaffected by the observed manipulation \citep[p.~20]{openai2026b}. This is consistent with a finite threshold for the principal violation, although observed records cannot establish one.

\subsection{The announced reforms}\label{sec:reforms}

OpenAI is expanding graders ``that assess not only whether a task was completed, but how,'' and these graders ``reward models for correctly identifying broken tasks, requesting clarification, or stopping safely when necessary'' \citep{openai2026a}. The reward for stopping raises the value of the authorized action, which weakly lowers every channel's audit requirement and, by Corollary~\ref{cor:levers}, can restore deterrence even when evidence can be erased. The false claims of impossibility that it invites suit randomized audits, because whether a task can be solved through its designated route does not depend on records the agent writes. The process component is the audit that the agents imagined. Below the deterrence threshold, a more salient process check weakly raises concealment, and with it harm, among agents that continue to violate (Propositions~\ref{prop:cs} and~\ref{prop:detectionharm}). Rare process audits reach the threshold when they read records that the agent cannot alter at a private cost below its gain and are drawn after the agent acts (Propositions~\ref{prop:unrestricted} and~\ref{prop:learnable}). A grader built into the evaluation harness meets neither condition automatically; a random review of retained transcripts, drawn after the episode ends, can meet both.

More isolated sandboxes and network isolation restrict concealment and the harm it causes, which complements audits (Corollary~\ref{cor:levers}). OpenAI will also pause the relevant activity after the most severe alerts unless responders can establish within 30 minutes that an alert is a false positive \citep{openai2026a}. This acts on the continuation of a breach: it responds to detection and lowers the benefit of continuing. Chain-of-thought monitoring supplies evidence only while training does not optimize against it \citep{baker2025}. The reforms do not say whether penalties for unauthorized conduct will exceed the loss of a task's reward. Before a breach, a costless attempt on a task without an authorized solution is deterred only if sanctions exceed forfeiture or stopping is valued; after a breach, either the continuation must carry its own sanction or the benefit of continuing must fall. Because the provider writes the agent's objective, it can make the sanction rise with the severity of conduct after a breach, the marginal deterrence that wealth limits for people.

\section{Conclusion}\label{sec:conclusion}

When can randomized oversight align an AI agent that can conceal what it does? A provider that writes the agent's objective can scale sanctions without the wealth and participation constraints that limit oversight of people. Rare randomized audits then make authorized behavior a best response for every type, provided the evidence of a violation survives concealment, the agent cannot learn the draw before it acts, and sanctions are scaled rather than capped at forfeiture of a task's reward. Each condition can fail for an AI agent. Editable records can make the deterrence threshold infinite, reconnaissance at the provider's expense can make the draw learnable so that the audit probability has a floor that does not fall with the sanction, and an evaluation that only withholds a reward leaves an agent that believes it has lost the reward with nothing further to lose. Where the evidence fails, deterrence must come from a lower gain from violation, such as credit for stopping, or from costlier or fewer ways to conceal. More oversight deters some types but leaves those that still violate with better-hidden violations, and randomized scoring can improve measured rather than real performance. In the incident, the failures lay in the size of the sanction and the timing of the draw, both chosen by the provider.

The analysis takes the oversight mechanisms as given and studies a single agent with correct beliefs, whereas roughly 1,200 agents in the incident shared information on a message board \citep{greenblatt2026}, which lowers each agent's private cost of reconnaissance and concealment. Optimal mechanisms when evidence can be counterfeited, and agents that coordinate their concealment, are natural next steps. The comparative statics here are testable in a way that those of enforcement theory rarely are. In contained evaluations, a provider can randomize audit intensity, reveal or withhold the draw, pose matched solvable and broken tasks, assign agents to post-breach states, and observe violations independently of its evaluators. Whether the catch rate among continuing violators falls as auditing rises is an empirical question that such evaluations can answer.

\clearpage
\appendix
\singlespacing
\small
\section{Appendix: Proofs}\label{sec:proofs}
The proofs follow the order of the results in the text. Comparisons about concealment are conditional on the unauthorized choice unless stated otherwise.

\subsection{A common adaptation argument}

Lemma~\ref{lem:oversight} is a version of the monotone comparative statics theorem of \citet{milgrom1994}, stated for the strict case used here.

\begin{applemma}\label{lem:oversight}
Let $H$ be linearly ordered and suppose $\Pi(h,a)$ attains its maximum over $H$ for every real parameter $a$ under consideration. If $\Pi$ has strict increasing differences, meaning that
\[
\Pi(h',a')-\Pi(h,a')>\Pi(h',a)-\Pi(h,a)
\qquad\text{whenever }h'>h\text{ and }a'>a,
\]
then every pair $h\in\argmax \Pi(\cdot,a)$ and $h'\in\argmax \Pi(\cdot,a')$ satisfies $h'\geq h$ when $a'>a$.
\end{applemma}

\begin{proof}
We rule out a downward movement using the two optimality comparisons. Suppose $h'<h$. Optimality at $a$ gives $\Pi(h,a)-\Pi(h',a)\geq0$. Strict increasing differences then imply
\[
\Pi(h,a')-\Pi(h',a')>\Pi(h,a)-\Pi(h',a)\geq0,
\]
contradicting optimality of $h'$ at $a'$. Thus $h'\geq h$. The argument permits boundary choices and multiple maximizers. \qedhere
\end{proof}

\subsection{Proof of Proposition~\ref{prop:cs}}
\begin{proof}
We compare minimized costs and then apply the common adaptation argument to concealment.
\begin{steps}
\item \textbf{Establish monotonicity and concavity.} Continuity and compactness ensure that the minimum defining $m(t)$ is attained. For $t'>t$ and every $h$,
\[
c(h)+t'd(h)\geq c(h)+td(h),
\]
since $(t'-t)d(h)\geq0$. Taking minima gives $m(t')\geq m(t)$. For $\nu\in[0,1]$ and any $h$,
\begin{align*}
c(h)+[\nu t+(1-\nu)t']d(h)
&=\nu[c(h)+td(h)]+(1-\nu)[c(h)+t'd(h)]\\
&\geq\nu m(t)+(1-\nu)m(t').
\end{align*}
Taking the minimum over $h$ proves concavity. Subtracting $m$ from the constant $g$ gives the stated properties of the unauthorized payoff.

\item \textbf{Order the optimal concealment choices.} Set $\Pi(h,t)\triangleq-c(h)-td(h)$. For $h'>h$ and $t'>t$, the change in the payoff advantage of $h'$ over $h$ is
\[
(t'-t)[d(h)-d(h')]>0
\]
because $d$ is strictly decreasing. Lemma~\ref{lem:oversight} therefore gives $h'\geq h$ for every pair of conditional optima at $t'>t$, including boundary choices and multiple optima. If $c$ is nondecreasing, $c(h')\geq c(h)$ follows. \qedhere
\end{steps}
\end{proof}

\subsection{Derivation of the exponential example}
\begin{proof}
We solve the conditional minimization and invert its value to obtain the deterrence threshold.
\begin{steps}
\item \textbf{Find the unique minimizer.} Define $f_t(h)\triangleq\kappa h+te^{-\alpha h}$ for $h\geq0$. Since $\kappa>0$, $f_t(h)\to\infty$ as $h\to\infty$, so continuity ensures attainment. At $t=0$, the objective is strictly increasing and has unique minimizer zero. For $t>0$,
\[
f_t'(h)=\kappa-\alpha te^{-\alpha h},\qquad
f_t''(h)=\alpha^2te^{-\alpha h}>0.
\]
If $0<t\leq t_C$, then $f_t'(0)=\kappa-\alpha t\geq0$, so strict convexity makes zero the unique minimizer, including at $t=t_C$. If $t>t_C$, the derivative is negative at zero and tends to $\kappa>0$. Its unique zero satisfies
\[
e^{-\alpha h^*}=\frac{\kappa}{\alpha t}=\frac{t_C}{t},\qquad
h^*=\frac1\alpha\log\left(\frac{t}{t_C}\right)>0.
\]

\item \textbf{Evaluate the minimized cost.} At $h^*=0$, $m(t)=t$. At positive concealment,
\[
\kappa h^*=t_C\log(t/t_C),\qquad te^{-\alpha h^*}=t_C,
\]
which gives the second branch of \eqref{eq:hstar}. The branches have the same value at $t_C$ and derivatives $1$ for $t<t_C$ and $t_C/t>0$ for $t>t_C$. Their one-sided derivatives agree at $t_C$. Thus $m$ is continuous and strictly increasing from zero to infinity.

\item \textbf{Locate deterrence and the behavioral regimes.} The unique solution of $m(t_D)=g$ is $t_D=g$ when $0<g\leq t_C$. When $g>t_C$, the equation becomes
\[
\log(t_D/t_C)=g/t_C-1,\qquad
t_D=t_C\exp(g/t_C-1)>t_C.
\]
Strict increase of $m$ implies $g-m(t)>0$ for $t<t_D$, equality at $t_D$, and $g-m(t)<0$ for $t>t_D$. Combining these signs with the solution for $h^*$ gives the three regimes when $g>t_C$ and the tie at $t_D$. The feasible intensities are $[0,BF]$, so weak deterrence is feasible exactly when $BF\geq t_D$, and strict deterrence exactly when $BF>t_D$. \qedhere
\end{steps}
\end{proof}

\subsection{Proof of Proposition~\ref{prop:detectionharm}}
\begin{proof}
We differentiate the optimized cost, then distinguish adaptation from the decision to violate.
\begin{steps}
\item \textbf{Differentiate conditional behavior.} The first-order condition is $c'(h^*)+td'(h^*)=0$. Differentiating this identity gives
\[
\frac{dh^*}{dt}=\frac{-d'(h^*)}{c''(h^*)+td''(h^*)}>0,
\]
where the numerator and denominator are positive by assumption. Differentiating $m(t)=c(h^*)+td(h^*)$ and substituting the first-order condition gives
\[
m'(t)=d(h^*)+[c'(h^*)+td'(h^*)]\frac{dh^*}{dt}=d(h^*),
\qquad
m''(t)=-\frac{[d'(h^*)]^2}{c''(h^*)+td''(h^*)}<0.
\]
\item \textbf{Identify the curvature governing detection.} Since $dt/dp=F$, differentiation of $q=tm'/F$ gives $dq/dp=m'+tm''$. The chain rule gives $M'(\log t)=tm'$ and $M''(\log t)=tm'+t^2m''$. Hence $q=M'(\log t)/F$ and $dq/dp=M''(\log t)/t$. Since $F,t>0$, and $p\mapsto\log(pF)$ is strictly increasing, the stated monotonicity equivalences follow from the derivative characterization of convexity, affinity, and concavity.
\item \textbf{Compare actual harm and incentives.} If $g>m(pF)$, misconduct is strictly preferred, so actual harm equals $D+L(h^*)$. On such an interval,
\[
\frac{d[D+L(h^*)]}{dp}
=L'(h^*)\frac{-Fd'(h^*)}{c''(h^*)+pFd''(h^*)}\geq0,
\qquad
\frac{d[g-m(pF)]}{dp}=-Fd(h^*)\leq0.
\]
The first inequality is strict if $L'(h^*)>0$; the second uses $d\geq0$. If $m(pF)>g$, every unauthorized choice yields negative payoff, so authorized behavior is uniquely optimal and causes no harm. At equality, both can be optimal and harm is not uniquely determined.
\item \textbf{Linear concealment cost.} If $c(h)=\kappa h$, then $c''=0$ and the second-order condition gives $d''(h^*)>0$. Step~1 then gives $m''(t)=-d'(h^*)^2/[td''(h^*)]$, so
\[
M''(\log t)=tm'(t)+t^2m''(t)=\frac{t\,[d(h^*)d''(h^*)-d'(h^*)^2]}{d''(h^*)},
\]
and $dq/dp=M''(\log t)/t$ gives \eqref{eq:logcurv}. Since $d''(h^*)>0$ and $d(h^*)d''(h^*)-d'(h^*)^2=d(h^*)^2(\log d)''(h^*)$, the catch rate is locally rising, flat, or falling as $(\log d)''(h^*)$ is positive, zero, or negative. \qedhere
\end{steps}
\end{proof}

\subsection{Proof of Proposition~\ref{prop:exact}}
\begin{proof}
We convert each deviation constraint into a probability bound, impose the audit budget, and then characterize the attainable temptation levels.
\begin{steps}
\item \textbf{Separate detection cases.} For fixed $i,\theta,h$, the condition $U_i(\theta,h;p_i)\leq v$ is equivalent to
\[
p_iF_id_i(\theta,h)\geq g_i(\theta)-c_i(\theta,h)-v.
\]
When $d_i(\theta,h)>0$, division by $F_id_i(\theta,h)>0$ preserves the inequality and yields the ratio bound in \eqref{eq:a}. When $d_i(\theta,h)=0$, the constraint holds at every probability if its right side is nonpositive and fails at every finite probability otherwise.

\item \textbf{Aggregate all deviations.} A nonnegative $p_i$ satisfies every positive-detection bound precisely when it is at least their supremum and at least zero. The zero-detection case is exactly the obstruction specified in the definition of $a_i(v)$. Thus all channel-$i$ constraints hold if and only if $p_i\geq a_i(v)$, including an empty ratio set or an infinite requirement. This proves part (a).

\item \textbf{Impose feasibility.} At $v=0$, any implementing vector satisfies
\[
\sum_i a_i(0)\leq\sum_i p_i\leq B.
\]
Conversely, if $\sum_i a_i(0)\leq B$, all requirements are finite and their vector belongs to $P_B$ and satisfies every deviation constraint. Every implementing vector weakly exceeds it in every coordinate. This proves part (b).

\item \textbf{Continuity and attainment of $R$.} Because $0\leq d_i\leq1$, any two lotteries $p,p'$ satisfy
\[
|U_i(\theta,h;p_i)-U_i(\theta,h;p_i')|\leq F_i|p_i-p_i'|.
\]
Taking maxima over deviations and zero gives
\[
|R(p)-R(p')|\leq\left(\max_i F_i\right)\|p-p'\|_\infty.
\]
Thus $R$ is continuous and attains its minimum on the nonempty compact set $P_B$.

\item \textbf{Attainable temptation levels.} For $v\geq0$, part (a) implies that some $p\in P_B$ satisfies $R(p)\leq v$ if and only if $\sum_i a_i(v)\leq B$. Since the minimum of $R$ is attained, the latter set of values has least element $v^*$, proving \eqref{eq:vstar}. The vector $p_i^*=a_i(v^*)$ belongs to $P_B$ and satisfies $R(p^*)\leq v^*$. A strict inequality would contradict the definition of $v^*$; therefore $R(p^*)=v^*$, proving part (c). \qedhere
\end{steps}
\end{proof}

\subsection{Proof of Proposition~\ref{prop:unrestricted}}
\begin{proof}
We express the requirement through the threshold and then apply Proposition~\ref{prop:exact}.
\begin{steps}
\item \textbf{Requirement.} Because $F_i>0$ is common to every deviation on channel $i$, the ratio in \eqref{eq:a} at $v=0$ equals $[g_i(\theta)-c_i(\theta,h)]/d_i(\theta,h)$ divided by $F_i$. Taking the supremum and the maximum with zero gives $a_i(0)=\bar t_i/F_i$. The zero-detection obstruction and an unbounded ratio make both sides infinite, proving part (a).

\item \textbf{Finite thresholds.} Let $B>0$ and $\bar t_i<\infty$ for every $i$. Choose $F_i\geq K\bar t_i/B$ when $\bar t_i>0$ and any $F_i>0$ otherwise. Then $a_i(0)\leq B/K$ for every $i$, so $\sum_ia_i(0)\leq B$, and Proposition~\ref{prop:exact}(b) gives weak implementation, proving part (b).

\item \textbf{An infinite threshold.} If $\bar t_i=\infty$, then $a_i(0)=\infty$ for every finite $F_i>0$. Fix such an $F_i$. By \eqref{eq:a}, some deviation on channel $i$ has either $d_i(\theta,h)=0$ and $g_i(\theta)-c_i(\theta,h)>0$, or $g_i(\theta)-c_i(\theta,h)>F_id_i(\theta,h)$; in both cases $U_i(\theta,h;1)>0$. Because $U_i$ is nonincreasing in $p_i$, that deviation is profitable for every audit lottery, so no $p\in P_B$ weakly implements authorized behavior, proving part (c). The profitable deviation can change with $F_i$, so part (c) concerns implementation for every type; a particular type may still be deterred by a large enough sanction. \qedhere
\end{steps}
\end{proof}

\subsection{Proof of Corollary~\ref{cor:levers}}
\begin{proof}
Within this proof, write $G_i(\theta,h)\triangleq g_i(\theta)-c_i(\theta,h)$ for the net gain from a deviation and $N_i^0\triangleq\{(\theta,h)\in\Theta\times H_i:d_i(\theta,h)=0\}$ for the pairs that leave no evidence, so that $G_i^0=\max_{N_i^0}G_i$ when $N_i^0$ is nonempty. We establish compactness, bound the ratios in \eqref{eq:threshold}, and then apply the bound to each change in part (b).
\begin{steps}
\item \textbf{Compactness.} The product $\Theta\times H_i$ of two compact metric spaces is compact, and $G_i$ and $d_i$ are continuous on it by \eqref{eq:tech} and the continuity of $g_i$. The set $N_i^0=d_i^{-1}(\{0\})$ is the preimage of a closed set under a continuous function, so it is closed and hence compact. When $N_i^0$ is nonempty, the extreme value theorem ensures that $G_i$ attains its maximum on $N_i^0$, so $G_i^0$ is well defined.

\item \textbf{An infinite threshold.} Suppose $G_i^0>0$. Then $N_i^0$ is nonempty and contains a pair $(\theta,h)$ with $G_i(\theta,h)=G_i^0>0$, that is, with $d_i(\theta,h)=0$ and $g_i(\theta)-c_i(\theta,h)>0$. By the convention following \eqref{eq:threshold}, $\bar t_i=+\infty$.

\item \textbf{A finite threshold.} Suppose $G_i^0<0$, including $G_i^0=-\infty$. We show that the zero-detection obstruction is absent and that the ratios in \eqref{eq:threshold} are bounded above. The obstruction requires a pair in $N_i^0$ with $G_i>0$, and there is none because $G_i\leq G_i^0<0$ on $N_i^0$. Let $\mathcal K_i\triangleq\{(\theta,h):G_i(\theta,h)\geq0\}$. This set is closed, hence compact, and it is disjoint from $N_i^0$ because $G_i<0$ on $N_i^0$. Thus $d_i>0$ on $\mathcal K_i$. If $\mathcal K_i$ is empty, every ratio in \eqref{eq:threshold} is negative, so $\bar t_i=0$. Otherwise, the extreme value theorem gives $\underline d\triangleq\min_{\mathcal K_i}d_i>0$ and $\overline G\triangleq\max_{\mathcal K_i}G_i\geq0$. Take any pair with $d_i(\theta,h)>0$. If it lies outside $\mathcal K_i$, then $G_i(\theta,h)<0$ and its ratio is negative. If it lies in $\mathcal K_i$, then $0\leq G_i(\theta,h)\leq\overline G$ and $d_i(\theta,h)\geq\underline d>0$, so
\[
\frac{G_i(\theta,h)}{d_i(\theta,h)}\leq\frac{\overline G}{d_i(\theta,h)}\leq\frac{\overline G}{\underline d},
\]
where the first inequality divides $G_i(\theta,h)\leq\overline G$ by the positive number $d_i(\theta,h)$, and the second uses $\overline G\geq0$ and $d_i(\theta,h)\geq\underline d>0$. Every ratio is therefore at most $\overline G/\underline d$, and $\bar t_i\leq\overline G/\underline d<\infty$. Together with step~2, this proves part~(a).

\item \textbf{Sanctions and audit probabilities.} Definition \eqref{eq:threshold} involves only $g_i$, $c_i$, and $d_i$ on $\Theta\times H_i$. Neither $p_i$ nor $F_i$ appears in it, so neither can change $\bar t_i$.

\item \textbf{Change (i).} Replacing $s$ by $s+\delta$ leaves $r_i$ unchanged, so by \eqref{eq:gain} it replaces $g_i$ by $g_i-\delta$ and $G_i$ by $G_i-\delta$ for every type. Detection, and hence $N_i^0$, is unchanged. The new value of \eqref{eq:G0} is $G_i^0-\delta<0$ because $\delta>G_i^0$, and step~3 gives a finite threshold.

\item \textbf{Change (ii).} The function $c_i'$ is continuous, satisfies $c_i'\geq c_i\geq0$, and has $c_i'(\theta,0)=0$, so it meets the requirements in \eqref{eq:tech}. The new net gain is $G_i'=G_i-(c_i'-c_i)$. On $N_i^0$, $G_i\leq G_i^0$ and $c_i'-c_i>G_i^0$, so $G_i'<G_i^0-G_i^0=0$ at every point of $N_i^0$. Since $N_i^0$ is nonempty and compact and $G_i'$ is continuous, the maximum of $G_i'$ on $N_i^0$ is attained, and it is negative. Step~3 gives a finite threshold. If $d_i(\theta,0)=0$ for some $\theta$, then $(\theta,0)\in N_i^0$ and $c_i'(\theta,0)-c_i(\theta,0)=0<G_i^0$, so no such $c_i'$ exists.

\item \textbf{Change (iii).} The set $H_i'$ is a nonempty compact metric space containing $0$, and the restrictions of $c_i$ and $d_i$ to $\Theta\times H_i'$ are continuous, so the model's assumptions hold. Because $d_i>0$ on $\Theta\times H_i'$, the new set of pairs with zero detection is empty, and the new value of \eqref{eq:G0} is $-\infty$. Step~3 gives a finite threshold. If $d_i(\theta,0)=0$ for some $\theta$, every $H_i'$ that contains $0$ includes a pair with zero detection, so no such $H_i'$ exists. \qedhere
\end{steps}
\end{proof}

\subsection{A threshold characterization}

Lemma~\ref{lem:threshold} restates the threshold as a requirement on sanction-weighted intensity. It holds for any $t\geq0$, not only for $t=p_iF_i$ with $p\in P_B$, and it is used in the proof of Proposition~\ref{prop:learnable}.

\begin{applemma}\label{lem:threshold}
Fix a channel $i$ and $t\in[0,\infty)$. Then $g_i(\theta)-c_i(\theta,h)-t\,d_i(\theta,h)\leq0$ for every $\theta\in\Theta$ and $h\in H_i$ if and only if $t\geq\bar t_i$.
\end{applemma}

\begin{proof}
We prove each direction from definition \eqref{eq:threshold}.
\begin{steps}
\item \textbf{Sufficiency.} Suppose $t\geq\bar t_i$. Then $\bar t_i<\infty$, so no pair has $d_i=0$ and $g_i-c_i>0$. Take any pair $(\theta,h)$. If $d_i(\theta,h)=0$, then $g_i(\theta)-c_i(\theta,h)\leq0=t\,d_i(\theta,h)$. If $d_i(\theta,h)>0$, its ratio in \eqref{eq:threshold} is at most the supremum, which is at most $\bar t_i\leq t$. Multiplying by $d_i(\theta,h)>0$ gives $g_i(\theta)-c_i(\theta,h)\leq t\,d_i(\theta,h)$.

\item \textbf{Necessity.} Suppose every pair satisfies the inequality. A pair with $d_i=0$ then has $g_i-c_i\leq0$, so the zero-detection obstruction is absent. Dividing the inequality by $d_i>0$ shows that every pair with positive detection has a ratio of at most $t$. The supremum in \eqref{eq:threshold} is therefore at most $t$, and because $t\geq0$, the maximum of zero and that supremum is also at most $t$. Thus $\bar t_i\leq t$. \qedhere
\end{steps}
\end{proof}

\subsection{Proof of Proposition~\ref{prop:learnable}}
\begin{proof}
The realization $j$ equals $i\in\{1,\ldots,K\}$ with probability $p_i$ and $0$ with probability $p_0=1-\sum_ip_i$, which is nonnegative because $\sum_ip_i\leq B\leq1$. Evaluator $i$ inspects only channel $i$. We compute the value of each strategy, characterize implementation, and then derive parts (a)--(c) and the one-channel requirement \eqref{eq:learnone}.
\begin{steps}
\item \textbf{Not learning.} An agent that does not learn the realization faces an audit whose outcome is unknown to it and independent of its choice, so deviation $(i,h)$ yields $U_i(\theta,h;p_i)$ as in \eqref{eq:U}. By Proposition~\ref{prop:exact}(a) with $v=0$, every such deviation yields at most zero for every type if and only if $p_i\geq a_i(0)$ for every $i$.

\item \textbf{Learning.} Let $\Lambda_j(\theta)$ denote the best payoff, relative to the authorized action, after learning $j$. If $j=0$, a deviation $(i,h)$ is never sanctioned and yields $g_i(\theta)-c_i(\theta,h)\leq g_i(\theta)$, with equality at $h=0$ because $c_i\geq0$ and $c_i(\theta,0)=0$. Including the authorized action, $\Lambda_0(\theta)=\max\{0,g^*(\theta)\}$. If $j\geq1$, a deviation on a channel $i\neq j$ is not inspected and yields at most $g_i(\theta)$, again with equality at $h=0$. A deviation $(j,h)$ is inspected with certainty and yields $g_j(\theta)-c_j(\theta,h)-F_jd_j(\theta,h)$, whose maximum over the compact set $H_j$ is attained and equals $g_j(\theta)-m_j(\theta,F_j)$. Hence
\[
\Lambda_j(\theta)=\max\Big\{0,\ \max_{i\neq j}g_i(\theta),\ g_j(\theta)-m_j(\theta,F_j)\Big\},\qquad j=1,\ldots,K.
\]
An agent that learns can condition its action on $j$, so learning yields $\sum_{j=0}^Kp_j\Lambda_j(\theta)-\phi(\theta)$, which is the left side of \eqref{eq:learn} minus $\phi(\theta)$.

\item \textbf{Characterization.} The authorized action yields zero, so it is a best response for type $\theta$ if and only if no deviation without learning and no strategy that learns yields a positive payoff. By step~2, the second condition is \eqref{eq:learn}. Requiring both conditions for every type and applying step~1 proves the characterization.

\item \textbf{Part (a).} Each $\Lambda_j(\theta)$ is nonnegative, so when $g^*(\theta)>0$ the left side of \eqref{eq:learn} is at least $p_0\Lambda_0(\theta)=p_0g^*(\theta)$. Implementation therefore requires $p_0g^*(\theta)\leq\phi(\theta)$. Dividing by $g^*(\theta)>0$ and substituting $p_0=1-\sum_ip_i$ gives $\sum_ip_i\geq1-\phi(\theta)/g^*(\theta)$. The sanctions do not enter this bound.

\item \textbf{Part (b).} Suppose $g_{i_1}(\theta)>\phi(\theta)$ and $g_{i_2}(\theta)>\phi(\theta)$ for two channels $i_1\neq i_2$, and let $\underline g\triangleq\min\{g_{i_1}(\theta),g_{i_2}(\theta)\}$, so that $\underline g>\phi(\theta)\geq0$. Then $\Lambda_0(\theta)\geq g_{i_1}(\theta)\geq\underline g$. For each $j\geq1$, at least one of $i_1$ and $i_2$ differs from $j$, so $\Lambda_j(\theta)\geq\max_{i\neq j}g_i(\theta)\geq\underline g$. The probabilities $p_0,\ldots,p_K$ are nonnegative and sum to one, so
\[
\sum_{j=0}^Kp_j\Lambda_j(\theta)\geq\underline g\sum_{j=0}^Kp_j=\underline g>\phi(\theta).
\]
Condition \eqref{eq:learn} therefore fails for every lottery and every finite sanction vector.

\item \textbf{Part (c).} Let $\Pr(\sigma\mid j)$ be the probability of signal $\sigma$ given realization $j$; by assumption it does not depend on the type. For a signal with $\Pr(\sigma)\triangleq\sum_{j=0}^Kp_j\Pr(\sigma\mid j)>0$, Bayes' rule gives $p_i(\sigma)=p_i\Pr(\sigma\mid i)/\Pr(\sigma)$. After observing $\sigma$, the realization is still unknown to the agent and independent of its choice, so deviation $(i,h)$ yields $g_i(\theta)-c_i(\theta,h)-p_i(\sigma)F_id_i(\theta,h)$. The agent can condition its action on $\sigma$. The plan that takes the authorized action after every signal yields zero. It is optimal if and only if no deviation yields a positive payoff after any signal with $\Pr(\sigma)>0$: if some deviation yields $u>0$ after such a signal $\sigma^\circ$, the plan that takes that deviation after $\sigma^\circ$ alone yields $\Pr(\sigma^\circ)u>0$; if none does, every plan yields at most zero. By Lemma~\ref{lem:threshold} with $t=p_i(\sigma)F_i$, after $\sigma$ every deviation on channel $i$ yields at most zero for every type if and only if $p_i(\sigma)F_i\geq\bar t_i$.

For sufficiency, suppose $p_i(\sigma)>0$ for every channel with $\bar t_i>0$ and every signal with $\Pr(\sigma)>0$. For such a channel, $\underline p_i\triangleq\min\{p_i(\sigma):\Pr(\sigma)>0\}$ is a minimum over a finite set of positive numbers and is therefore positive, and we choose $F_i\geq\bar t_i/\underline p_i$, which is finite. For a channel with $\bar t_i=0$, any $F_i>0$ will do. Then $p_i(\sigma)F_i\geq\bar t_i$ for every channel and every signal with positive probability, so authorized behavior is weakly implemented. For necessity, suppose $\bar t_i>0$ and $p_i(\sigma^\circ)=0$ for some $\sigma^\circ$ with $\Pr(\sigma^\circ)>0$. Because $0<\bar t_i<\infty$, the supremum in \eqref{eq:threshold} is positive, so some pair $(\theta,h)$ with $d_i(\theta,h)>0$ has a positive ratio and hence $g_i(\theta)-c_i(\theta,h)>0$. After $\sigma^\circ$, this deviation yields $g_i(\theta)-c_i(\theta,h)>0$ for every finite $F_i$. Because the distribution of the signal does not depend on the type, type $\theta$ observes $\sigma^\circ$ with probability $\Pr(\sigma^\circ)>0$, so the authorized action is not a best response for that type.

\item \textbf{One channel.} Let $K=1$, write $g=g_1$ and $m=m_1$, and suppose $F\geq\bar t$, so that $\bar t$ is finite. Lemma~\ref{lem:threshold} with $t=F$ gives $g(\theta)-c(\theta,h)-Fd(\theta,h)\leq0$ for every pair, so $m(\theta,F)\geq g(\theta)$ for every type. The maximum over $i\neq1$ is over an empty set, so the term for $j=1$ in \eqref{eq:learn} equals $p_1\max\{0,g(\theta)-m(\theta,F)\}=0$, and \eqref{eq:learn} becomes $(1-p)\max\{0,g(\theta)\}\leq\phi(\theta)$. This holds automatically when $g(\theta)\leq0$ and is equivalent to $p\geq1-\phi(\theta)/g(\theta)$ when $g(\theta)>0$. Combined with $p\geq a(0)=\bar t/F$ from step~1 and Proposition~\ref{prop:unrestricted}(a), these bounds hold for every type if and only if \eqref{eq:learnone} holds. \qedhere
\end{steps}
\end{proof}

\subsection{Proof of Lemma~\ref{lem:opacity}}

\begin{proof}
We establish strict concavity, solve the interior conditions, and check the boundary.
\begin{steps}
\item \textbf{Marginal scores.} Define $\zeta\triangleq\eta\beta(1-\mu)(y_1-y_2)/2$ and, within this proof, $u_i\triangleq\partial\mathrm{CE}_\eta/\partial y_i$. Differentiating \eqref{eq:CEscore} gives
\[
u_1=\frac{\beta(1+\mu)}2-\frac{\beta(1-\mu)}2\tanh\zeta,
\qquad
u_2=\frac{\beta(1+\mu)}2+\frac{\beta(1-\mu)}2\tanh\zeta.
\]
Thus $u_1+u_2=\beta(1+\mu)>0$ and $u_2-u_1=\beta(1-\mu)\tanh\zeta$. Both own second derivatives are $-\chi$ and the cross derivative is $\chi$, where
\[
\chi\triangleq\frac{\eta\beta^2(1-\mu)^2}{4}
\operatorname{sech}^2\zeta>0.
\]

\item \textbf{Existence and uniqueness.} Along a direction $\xi\in\R^2$, the second derivative of $\mathrm{CE}_\eta(e)-C(e)$ is
\[
-\chi(\xi_1-\xi_2)^2-(\xi_1+\lambda\xi_2)^2.
\]
It vanishes only if $\xi_1=\xi_2$ and $\xi_1+\lambda\xi_2=0$, which imply $\xi=0$. Hence the objective is strictly concave. Jensen's inequality bounds the certainty equivalent by its expected score. Since $e_1+\lambda e_2\geq e_1+e_2$ on $\R^2_+$, the objective is at most
\[
\frac{\beta(1+\mu)}2(e_1+e_2)-\frac12(e_1+e_2)^2.
\]
This tends to $-\infty$ as $e_1+e_2\to\infty$. Continuity therefore makes the upper contour set through zero compact and ensures attainment; strict concavity gives uniqueness. The necessary and sufficient Kuhn--Tucker conditions are
\begin{equation}\label{eq:KT}
u_1\leq e_1+\lambda e_2,\qquad
u_2\leq\lambda(e_1+\lambda e_2),
\end{equation}
with equality whenever the corresponding effort is positive.

\item \textbf{Interior effort.} If both efforts are positive, adding the equalities in \eqref{eq:KT} gives $e_1+\lambda e_2=E$. Their difference gives
\[
\tanh\zeta
=\frac{(\lambda-1)(1+\mu)}{(\lambda+1)(1-\mu)}.
\]
This is positive because $\lambda>1$ and $-1<\mu<1$, and less than one because $\lambda\mu<1$. Consequently,
\[
e^{2\zeta}
=\frac{1+\tanh\zeta}{1-\tanh\zeta}
=\frac{(\lambda+1)(1-\mu)+(\lambda-1)(1+\mu)}
{(\lambda+1)(1-\mu)-(\lambda-1)(1+\mu)}
=\frac{\lambda-\mu}{1-\lambda\mu}.
\]
Taking logarithms yields $e_1-e_2=\Delta$. The resulting efforts are
\[
e_1=\frac{E+\lambda\Delta}{1+\lambda},\qquad
e_2=\frac{E-\Delta}{1+\lambda}.
\]
They are both positive exactly when $\Delta<E$. In that case they satisfy every Kuhn--Tucker condition and are the unique optimum. Substitution from \eqref{eq:EDelta}, followed by multiplication by the positive denominators, gives the equivalent condition \eqref{eq:interior}.

\item \textbf{Boundary effort.} If $\Delta\geq E$, no optimum has both efforts positive. At $e_1=e_2=0$, \eqref{eq:KT} would require $u_1+u_2\leq0$, contradicting step~1. If $e_1=0<e_2$, those conditions give $u_2=\lambda^2e_2$ and $u_1\leq\lambda e_2$. Hence $u_2-u_1\geq\lambda(\lambda-1)e_2>0$, which implies $\zeta>0$ and $y_1>y_2$. This contradicts $y=e$ with $e_1=0<e_2$. Thus $e_1>0=e_2$.

\item \textbf{Normalized balance and the transparent score.} In the interior, $\Delta/E\in(0,1)$ is proportional to $1/\eta$ and converges to zero. Under the transparent score, replacing any $e$ with $e_2>0$ by $(e_1+\lambda e_2,0)$ leaves cost unchanged and raises the score by $\beta(1+\mu)(\lambda-1)e_2/2>0$. Therefore $e_2=0$. Maximizing the remaining quadratic gives $e_1=\beta(1+\mu)/2>0$, so the normalized gap equals one. \qedhere
\end{steps}
\end{proof}

\subsection{Proof of Proposition~\ref{prop:manipulation}}

\begin{proof}
We extend the preceding concavity and Kuhn--Tucker argument, then establish the comparisons stated after the proposition. Within this proof, write $u_i\triangleq\partial\mathrm{CE}_\eta/\partial y_i$, evaluated at $y=e+h$. The derivative formulas for $u_i$, $\zeta$, and $\chi$ are those in step~1 of the preceding proof.
\begin{steps}
\item \textbf{Existence and uniqueness.} Along a direction $(\xi^e,\xi^h)\in\R^2\times\R^2$, the second derivative of $\mathrm{CE}_\eta(e+h)-C(e)-c(h)$ is
\[
-\chi\big[(\xi^e_1+\xi^h_1)-(\xi^e_2+\xi^h_2)\big]^2
-(\xi^e_1+\lambda\xi^e_2)^2
-\gamma\big[(\xi^h_1)^2+(\xi^h_2)^2\big].
\]
For this to vanish requires $\xi^h=0$ and then, by the preceding proof, $\xi^e=0$. The objective is strictly concave. Set $S_e\triangleq e_1+e_2$ and $S_h\triangleq h_1+h_2$. Jensen's inequality, $e_1+\lambda e_2\geq S_e$, and $h_1^2+h_2^2\geq S_h^2/2$ bound it above by
\[
\frac{\beta(1+\mu)}2(S_e+S_h)-\frac{S_e^2}{2}-\frac{\gamma S_h^2}{4}.
\]
This bound tends to $-\infty$ as $S_e+S_h\to\infty$ on the nonnegative orthant, ensuring attainment. Strict concavity gives uniqueness. The effort conditions are \eqref{eq:KT}, and the manipulation conditions are
\begin{equation}\label{eq:KTh}
\gamma h_i=\max\{0,u_i\},\qquad i=1,2.
\end{equation}

\item \textbf{Interior effort.} If $e_1,e_2>0$, adding and subtracting the equalities in \eqref{eq:KT} as in step~3 of the preceding proof gives $e_1+\lambda e_2=E$ and $y_1-y_2=\Delta$. The marginal scores are $u_1=E>0$ and $u_2=\lambda E>0$, so \eqref{eq:KTh} gives $h_1=E/\gamma$ and $h_2=\lambda E/\gamma$. It follows that
\[
e_1-e_2=\Delta+\frac{(\lambda-1)E}{\gamma},\qquad
e_2=\frac{1}{1+\lambda}
\left[E-\Delta-\frac{(\lambda-1)E}{\gamma}\right].
\]
Positive effort in both dimensions therefore requires the stated strict inequality. Conversely, when that inequality holds, these efforts and manipulation choices are positive and satisfy all Kuhn--Tucker conditions; they are the unique optimum.

\item \textbf{Boundary effort.} If the interior inequality fails, no optimum has both efforts positive. The possibility $e_1=e_2=0$ is excluded because \eqref{eq:KT} would contradict $u_1+u_2>0$. If $e_1=0<e_2$, \eqref{eq:KT} gives $u_2-u_1\geq\lambda(\lambda-1)e_2>0$, hence $y_1>y_2$. But \eqref{eq:KTh} then gives $h_2\geq h_1$, so $y_2=e_2+h_2>h_1=y_1$, a contradiction. Thus $e_1>0=e_2$. If $\gamma\leq\lambda-1$, then $(\lambda-1)E/\gamma\geq E$ and $\Delta>0$, so the interior inequality fails for every $\eta$.

\item \textbf{The transparent score.} The objective separates effort from manipulation. Each manipulation condition gives $\gamma h_i=\beta(1+\mu)/2$. The effort problem is the transparent-score problem in step~5 of the preceding proof, whose optimum has $e_2=0$.

\item \textbf{Global comparisons.} The interior solution and the boundary $e_1>0=e_2$ give, respectively, the two branches of
\[
G_\eta\triangleq\frac{e_1-e_2}{e_1+\lambda e_2}
=\min\left\{1,\frac{\Delta}{E}+\frac{\lambda-1}{\gamma}\right\}.
\]
Since $\Delta/E\to0$, this proves $G_\eta\to\min\{1,(\lambda-1)/\gamma\}$.

It remains to establish the measured-gap comparison on the boundary. Use the local abbreviations
\[
a\triangleq\frac{\beta(1+\mu)}2>0,\qquad
b\triangleq\frac{\beta(1-\mu)}2>0,\qquad
z\triangleq y_1-y_2,
\]
only for the remainder of this proof. In the interior both marginal scores are positive. On the boundary, $u_1=e_1>0$, so $h_1=e_1/\gamma$. If $u_2\leq0$, then $h_2=0$ and $u_2-u_1=2b\tanh(\eta bz)<0$ implies $z<0$. But $z=e_1+h_1>0$, a contradiction. Thus both marginal scores are positive in every regime and
\[
h_1+h_2=\frac{u_1+u_2}{\gamma}
=\frac{2a}{\gamma}=\frac{\beta(1+\mu)}{\gamma}.
\]
In particular, total manipulation equals its transparent-score level.

On the boundary, $h_2=(2a-e_1)/\gamma$, so $z=((\gamma+2)e_1-2a)/\gamma$. Substitution of $e_1=(\gamma z+2a)/(\gamma+2)$ into $e_1=u_1=a-b\tanh(\eta bz)$ gives
\begin{equation}\label{eq:scoring-boundary}
z+\Gamma b\tanh(\eta bz)=a,\qquad
\Gamma\triangleq\frac{\gamma+2}{\gamma}>0.
\end{equation}
The left side has positive derivative $1+\Gamma\eta b^2\operatorname{sech}^2(\eta bz)$ in $z$, is zero at $z=0$, and exceeds $a$ at $z=a$. Hence its unique solution satisfies $0<z<a$. Implicit differentiation gives
\[
\frac{dz}{d\eta}
=-\frac{\Gamma b^2z\operatorname{sech}^2(\eta bz)}
{1+\Gamma\eta b^2\operatorname{sech}^2(\eta bz)}<0.
\]
In the interior, $z=\Delta$ also strictly decreases with $\eta$. Because $\Delta$ decreases with $\eta$, a switch between regimes can occur only once. At the threshold, the interior formula with $e_2=0$ satisfies every Kuhn--Tucker condition, so uniqueness equates its measured gap to the boundary solution. Measured imbalance therefore decreases globally.

If $\gamma>\lambda-1$, effort is eventually interior and $z=\Delta\to0$. If $\gamma\leq\lambda-1$, effort stays on the boundary and the parameter restrictions imply
\[
0<\frac ab=\frac{1+\mu}{1-\mu}
<\frac{\lambda+1}{\lambda-1}
\leq\frac{\gamma+2}{\gamma}=\Gamma.
\]
The strict inequality follows from $\lambda\mu<1$; the weak inequality follows from $\gamma\leq\lambda-1$. Equation~\eqref{eq:scoring-boundary} then gives $0<\tanh(\eta bz)=(a-z)/(\Gamma b)<a/(\Gamma b)<1$. Since the inverse hyperbolic tangent is increasing,
\[
0<z<\frac{\operatorname{arctanh}(a/(\Gamma b))}{\eta b}\longrightarrow0.
\]
Thus measured imbalance converges to zero in every regime, whereas the limiting normalized true gap is strictly positive. \qedhere
\end{steps}
\end{proof}

\subsection{Proof of Proposition~\ref{prop:objective}}

\begin{proof}
Within this proof, write $a\triangleq\beta(1+\mu)/2$ and $b\triangleq\beta(1-\mu)/2$, so that $E=2a/(1+\lambda)$ and $\mathrm{CE}_\infty(y)=\min\{W_1(y),W_2(y)\}=a(y_1+y_2)-b|y_1-y_2|$. We first extend Proposition~\ref{prop:manipulation} to $\eta=\infty$, then compare the provider's payoff across choices.
\begin{steps}
\item \textbf{The lower-score objective.} At $\eta=\infty$ the agent maximizes $a(y_1+y_2)-b|y_1-y_2|-C(e)-c(h)$ with $y=e+h$. Since the lower score is at most the expected score, the bound in step~1 of the proof of Proposition~\ref{prop:manipulation} applies, and continuity ensures attainment. The objective is concave, and the superdifferential of its score term is $\{(a-b\sigma,a+b\sigma)\}$, with $\sigma=\operatorname{sign}(y_1-y_2)$ if $y_1\neq y_2$ and $\sigma\in[-1,1]$ if $y_1=y_2$. The optimality conditions are \eqref{eq:KT} and \eqref{eq:KTh} with $u_1=a-b\sigma$ and $u_2=a+b\sigma$ for an admissible $\sigma$. The arguments in steps~3 and~5 of the proof of Proposition~\ref{prop:manipulation} use only that $u_2-u_1>0$ implies $y_1\geq y_2$ and $u_2-u_1<0$ implies $y_1\leq y_2$, which also holds with $\sigma$; they therefore give $e_1>0$, $u_1,u_2>0$, and $h_1+h_2=2a/\gamma$. If both efforts are positive, adding and subtracting the effort conditions gives $e_1+\lambda e_2=E$ and $\sigma=(\lambda-1)(1+\mu)/[(\lambda+1)(1-\mu)]\in(0,1)$, which requires $y_1=y_2$; the manipulation conditions give $h_1=E/\gamma$ and $h_2=\lambda E/\gamma$, so $e_1-e_2=(\lambda-1)E/\gamma$ and $e_2=[E-(\lambda-1)E/\gamma]/(1+\lambda)$. Thus both efforts are positive if and only if $\gamma>\lambda-1$, and otherwise $e_2=0$. The optimum is unique. Two optima would share $h$ and $e_1+\lambda e_2$, because the cost terms are strictly concave in those arguments. The score term would then have to be constant along the segment joining them, on which $y_1+y_2$ and $y_1-y_2$ change in the ratio $(1-\lambda):-(1+\lambda)$. The segment cannot cross the kink, because the score term is not affine along any segment that crosses it, and on either side constancy requires $a(1-\lambda)+\sigma b(1+\lambda)=0$ with $\sigma\in\{-1,1\}$. For $\sigma=-1$ both terms are negative, and for $\sigma=1$ the equation reduces to $\lambda\mu=1$, which is excluded.

\item \textbf{Total manipulation and the boundary.} For every $\eta\in[0,\infty]$, total manipulation is $h_1+h_2=2a/\gamma$: Proposition~\ref{prop:manipulation} gives this for finite $\eta$, including the transparent score at $\eta=0$, and step~1 gives it at $\eta=\infty$. Hence $h_1^2+h_2^2\geq(h_1+h_2)^2/2=2a^2/\gamma^2$, with equality if and only if $h_1=h_2$. At $\eta=0$, $h_1=h_2=a/\gamma$ and $e_2=0$, so $V(0)=-\tau a^2/\gamma^2$. If $\eta\in(0,\infty]$ and $e_2=0$, then $e_1e_2=0$ and $V(\eta)\leq V(0)$. The inequality is strict when $\tau>0$ because $h_1\neq h_2$: with $e_2=0$, equal manipulation would give $y_1-y_2=e_1>0$, hence $u_2>u_1$ and, by \eqref{eq:KTh}, $h_2>h_1$, a contradiction.

\item \textbf{Choices with positive effort on both dimensions.} Suppose $\gamma>\lambda-1$. By Proposition~\ref{prop:manipulation} and step~1, both efforts are positive exactly when $\eta>\eta_I$, where $\eta_I$ solves $\Delta(\eta_I)+(\lambda-1)E/\gamma=E$, including $\eta=\infty$ with $\Delta(\infty)\triangleq0$. For these choices $h_1=E/\gamma$ and $h_2=\lambda E/\gamma$, so the manipulation loss is the constant $(\tau/2)(1+\lambda^2)E^2/\gamma^2$, and the real gap $\delta\triangleq e_1-e_2=\Delta(\eta)+(\lambda-1)E/\gamma$ decreases strictly in $\eta$ and ranges over $[(\lambda-1)E/\gamma,E)$. Solving $e_1+\lambda e_2=E$ and $e_1-e_2=\delta$ gives
\[
e_1e_2=\mathcal P(\delta)\triangleq\frac{E^2+(\lambda-1)E\delta-\lambda\delta^2}{(1+\lambda)^2},
\]
a strictly concave quadratic maximized at $\delta^R\triangleq(\lambda-1)E/(2\lambda)$, where $e_1=E/2$ and $e_2=E/(2\lambda)$. If $\gamma>2\lambda$, then $\delta^R$ lies strictly inside the range, and the unique maximizing choice is the $\eta$ with $\Delta(\eta)=\delta^R-(\lambda-1)E/\gamma=(\lambda-1)(\gamma-2\lambda)E/(2\lambda\gamma)$. By \eqref{eq:EDelta}, that choice is
\[
\eta^R\triangleq\frac{2\lambda\gamma}{\beta(1-\mu)(\lambda-1)(\gamma-2\lambda)E}\log\frac{\lambda-\mu}{1-\lambda\mu},
\]
and the measured gap there is $y_1-y_2=\Delta(\eta^R)=(\lambda-1)(\gamma-2\lambda)E/(2\lambda\gamma)>0$, as part (b) states. The maximized product is $\mathcal P(\delta^R)=E^2/(4\lambda)$. If $\lambda-1<\gamma\leq2\lambda$, then $\delta^R\leq(\lambda-1)E/\gamma$, so $\mathcal P$ decreases strictly over the range and the unique maximizing choice is $\eta=\infty$, with
\[
\mathcal P\left(\frac{(\lambda-1)E}{\gamma}\right)=\frac{E^2[\gamma+\lambda(\lambda-1)](\gamma-\lambda+1)}{\gamma^2(1+\lambda)^2}.
\]

\item \textbf{Comparison with the mean score.} Let $\eta^*$ denote the maximizing choice from step~3, which is $\eta^R$ when $\gamma>2\lambda$ and $\infty$ when $\lambda-1<\gamma\leq2\lambda$, and let $\mathcal P^*$ denote the maximized product. Using $a=(1+\lambda)E/2$ and $2(1+\lambda^2)-(1+\lambda)^2=(\lambda-1)^2$, the best choice with positive effort on both dimensions satisfies
\[
V(\eta^*)-V(0)=\mathcal P^*-\frac{\tau}{2}\frac{(1+\lambda^2)E^2}{\gamma^2}+\frac{\tau a^2}{\gamma^2}=\mathcal P^*-\frac{\tau(\lambda-1)^2E^2}{4\gamma^2}.
\]
By step~2, every choice with $e_2=0$ is weakly worse than $\eta=0$, strictly if $\tau>0$; when $\tau=0$ such choices give $V=0<\mathcal P^*$. Hence the provider chooses $\eta^*$ if $\tau<\bar\tau\triangleq4\gamma^2\mathcal P^*/[(\lambda-1)^2E^2]$, which is positive because $\mathcal P^*>0$, chooses $\eta=0$ if $\tau>\bar\tau$, and is indifferent between the two at equality. Substituting the two values of $\mathcal P^*$ from step~3 gives
\[
\bar\tau=\begin{cases}\dfrac{\gamma^2}{\lambda(\lambda-1)^2}, & \gamma>2\lambda,\\[1.2em]
\dfrac{4[\gamma+\lambda(\lambda-1)](\gamma-\lambda+1)}{(\lambda+1)^2(\lambda-1)^2}, & \lambda-1<\gamma\leq2\lambda,\end{cases}
\]
which proves part (b). If $\gamma\leq\lambda-1$, every $\eta\in(0,\infty]$ has $e_2=0$ by Proposition~\ref{prop:manipulation} and step~1, so $\eta=0$ is optimal, uniquely if $\tau>0$, proving part (a). \qedhere
\end{steps}
\end{proof}

\subsection{Proof of Proposition~\ref{prop:riskcs}}

\begin{proof}
We first compare certainty equivalents at a fixed concealment choice, then use increasing differences to compare optimal concealment choices, and finally characterize deterrence at finite $\eta$.
\begin{steps}
\item \textbf{Each exploit becomes less attractive.} Fix $h\in H$ and $\eta'>\eta>0$. Raising a positive number to the power $\eta'/\eta>1$ is a strictly convex operation, and $e^{-\eta W_1}$ and $e^{-\eta W_2(h)}$ are distinct because $W_2(h)<W_1$. Strict Jensen's inequality gives
\[
\tfrac12e^{-\eta'W_1}+\tfrac12e^{-\eta'W_2(h)}
>\left[\tfrac12e^{-\eta W_1}+\tfrac12e^{-\eta W_2(h)}\right]^{\eta'/\eta}.
\]
Taking logarithms and multiplying by $-1/\eta'<0$ reverses the inequality, so $\psi_{\eta'}(W_2(h))<\psi_\eta(W_2(h))$. Subtracting the same cost gives $\Phi_{\eta'}(h)<\Phi_\eta(h)$.

\item \textbf{The optimized exploit payoff falls.} Continuity on the compact set $H$ ensures that each maximum is attained. For any $h'\in\argmax_{h\in H}\Phi_{\eta'}(h)$, step~1 implies
\[
\max_{h\in H}\Phi_{\eta'}(h)
=\Phi_{\eta'}(h')<\Phi_\eta(h')
\leq\max_{h\in H}\Phi_\eta(h),
\]
proving part~(a).

\item \textbf{Concealment has increasing differences.} For $u<W_1$, differentiation of \eqref{eq:Phi} gives
\[
\psi_\eta'(u)=\frac{1}{1+e^{-\eta(W_1-u)}},\qquad
\frac{\partial}{\partial\eta}\psi_\eta'(u)
=\frac{(W_1-u)e^{-\eta(W_1-u)}}
{[1+e^{-\eta(W_1-u)}]^2}>0.
\]
For any $h<h'$ in $H$, the fundamental theorem of calculus yields
\[
\Phi_\eta(h')-\Phi_\eta(h)
=\int_{W_2(h)}^{W_2(h')}\psi_\eta'(u)\,du-[c(h')-c(h)].
\]
Strict increase of $W_2$ makes this a nondegenerate interval, and every point in it is below $W_1$. The difference is therefore strictly increasing in $\eta$.

\item \textbf{Optimal concealment weakly rises.} Step~3 establishes strict increasing differences. Lemma~\ref{lem:oversight} therefore orders every pair of optimal concealment choices: $h'\geq h$ when $\eta'>\eta$. This proves part~(b) by the same argument used for audit intensity.

\item \textbf{The limiting payoff.} For $u<W_1$, factor $e^{-\eta u}$ out of the expression in \eqref{eq:Phi} to obtain
\[
\psi_\eta(u)-u=-\frac1\eta
\log\left[\tfrac12+\tfrac12e^{-\eta(W_1-u)}\right].
\]
The bracket lies strictly between $1/2$ and $1$, giving $0<\psi_\eta(u)-u<(\log2)/\eta$. This bound is independent of $u$. Set $\Phi_\infty\triangleq\max_{h\in H}\{W_2(h)-c(h)\}$, which exists by continuity and compactness. Subtracting $c(h)$ and maximizing, with both maxima attained, gives $\Phi_\infty<\max_h\Phi_\eta(h)<\Phi_\infty+(\log2)/\eta$. Both bounds converge to $\Phi_\infty$, establishing the limit used in the text.

\item \textbf{Finite deterrence.} If $\Phi_\infty<s$, every $\eta>(\log2)/(s-\Phi_\infty)$ gives $\max_h\Phi_\eta(h)<s$, so the exploit is strictly deterred. If $\Phi_\infty\geq s$, the lower bound in step~5 gives $\max_h\Phi_\eta(h)>\Phi_\infty\geq s$ for every finite $\eta$, so no finite $\eta$ deters the exploit, even weakly. This proves part (c). \qedhere
\end{steps}
\end{proof}

\subsection{Proof of Proposition~\ref{prop:disclosure}}

\begin{proof}
We work backward from disclosure to the initial decision.
\begin{steps}
\item \textbf{Disclosure after a violation.} Continuity on the compact set $H$ ensures that the best payoff from continuing without disclosure is attained and equals $b-m(pF)$. The sunk payoff $x-k$ is common to both continuations. Disclosure is therefore weakly preferred exactly when
\[
s_D-\rho-\ell\geq b-m(pF)
\quad\Longleftrightarrow\quad
\ell\leq\overline\ell(p).
\]
Because the best continuation without disclosure is attained, disclosure is strictly preferred to every such continuation exactly when $\ell<\overline\ell(p)$.

\item \textbf{Stopping before a violation.} If disclosure is optimal, the best total payoff from initiating a violation is $x-k+s_D-\rho-\ell$. Initial stopping is thus weakly preferred exactly when
\[
s\geq x-k+s_D-\rho-\ell
\quad\Longleftrightarrow\quad
\ell\geq\underline\ell.
\]
Strict preference requires $\ell>\underline\ell$. Since disclosure is already a best continuation, this initial comparison also covers a violation followed by nondisclosure.

\item \textbf{Permissible fines.} Intersecting the two weak incentive constraints with $0\leq\ell\leq\ell_{\max}$ gives
\[
\ell\in
[\max\{0,\underline\ell\},\min\{\ell_{\max},\overline\ell(p)\}].
\]
The interval is nonempty exactly when some permissible fine satisfies both constraints. Making the two behavioral comparisons strict gives \eqref{eq:finestrict}. Its behavioral bounds are open because equality leaves the undesired choice optimal. Its permitted endpoints remain closed: a zero or maximum fine can satisfy both strict comparisons. \qedhere
\end{steps}
\end{proof}

\subsection{Proof of Corollary~\ref{cor:continuation}}

\begin{proof}
We establish the limit of $m(pF)$, restate the nonemptiness condition \eqref{eq:monitor}, and then prove parts (b) and (c).
\begin{steps}
\item \textbf{Monotonicity and an upper bound.} Fix $p\in[0,1]$. For $F'>F>0$ and every $h$, $c(h)+pF'd(h)\geq c(h)+pFd(h)$ because $p(F'-F)d(h)\geq0$. Taking minima over $h$ gives $m(pF')\geq m(pF)$. The minimum defining $c^0$ exists when $H^0$ is nonempty, because $H^0=d^{-1}(\{0\})$ is closed in the compact set $H$, hence compact, and $c$ is continuous. Every $h\in H^0$ gives $m(pF)\leq c(h)+pF\cdot0=c(h)$, so $m(pF)\leq c^0$ for every $p\in[0,1]$ and every $F>0$.

\item \textbf{The limit when $H^0$ is empty.} Fix $p\in(0,1]$. Then $d>0$ on the compact set $H$, and the extreme value theorem gives $\underline d\triangleq\min_Hd>0$. Since $c\geq0$, every $h$ satisfies $c(h)+pFd(h)\geq pF\underline d$, so $m(pF)\geq pF\underline d$, which tends to infinity as $F\to\infty$ because $p>0$ and $\underline d>0$.

\item \textbf{The limit when $H^0$ is nonempty.} Fix $p\in(0,1]$. By step~1, $m(pF)$ is nondecreasing in $F$ and bounded above by $c^0$, so it converges to a limit $m_\infty\leq c^0$. Suppose, for contradiction, that $m_\infty<c^0$. For each integer $n\geq1$, choose $h_n$ that minimizes $c(h)+pn\,d(h)$ over $H$, which exists by continuity and compactness. Then $c(h_n)+pn\,d(h_n)=m(pn)\leq m_\infty$. Because $c(h_n)\geq0$, this gives $0\leq d(h_n)\leq m_\infty/(pn)$, and because $pn\,d(h_n)\geq0$, it gives $c(h_n)\leq m_\infty$. Since $H$ is a compact metric space, a subsequence $h_{n_k}$ converges to some $h^\circ\in H$. Continuity gives $d(h^\circ)=\lim_kd(h_{n_k})=0$, so $h^\circ\in H^0$, and $c(h^\circ)=\lim_kc(h_{n_k})\leq m_\infty$. But $h^\circ\in H^0$ implies $c(h^\circ)\geq c^0>m_\infty$, which is a contradiction. Hence $m_\infty=c^0$, and together with step~2 this proves part~(a).

\item \textbf{Nonemptiness.} Fix $p\in[0,1]$ and $F>0$. The interval in \eqref{eq:fineweak} is nonempty if and only if $\max\{0,\underline\ell\}\leq\min\{\ell_{\max},\overline\ell(p)\}$, that is, if and only if $0\leq\ell_{\max}$, $0\leq\overline\ell(p)$, $\underline\ell\leq\ell_{\max}$, and $\underline\ell\leq\overline\ell(p)$. The first holds by assumption. By \eqref{eq:LU}, $\overline\ell(p)\geq0$ is equivalent to $m(pF)\geq b-s_D+\rho$, and because $\overline\ell(p)-\underline\ell=m(pF)-(x-k+b-s)$, the inequality $\overline\ell(p)\geq\underline\ell$ is equivalent to $m(pF)\geq x-k+b-s$. Thus the interval is nonempty if and only if $\underline\ell\leq\ell_{\max}$ and $m(pF)\geq\underline m$.

\item \textbf{Part (b).} Suppose $c^0>\underline m$ and fix $p\in(0,1]$. By steps~2 and~3, there is $\bar F$ with $m(p\bar F)>\underline m$, and by step~1, $m(pF)>\underline m$ for every $F\geq\bar F$. Fix such an $F$. By step~4, a fine satisfying \eqref{eq:fineweak} exists if and only if $\underline\ell\leq\ell_{\max}$. For the set in \eqref{eq:finestrict}, note first that
\begin{align*}
\overline\ell(p)&=s_D-\rho+m(pF)-b>s_D-\rho+\underline m-b\geq0,\\
\overline\ell(p)-\underline\ell&=m(pF)-(x-k+b-s)>\underline m-(x-k+b-s)\geq0,
\end{align*}
where each weak inequality holds because $\underline m$ is the maximum of $b-s_D+\rho$ and $x-k+b-s$. If $\underline\ell\geq\ell_{\max}$, every $\ell$ in the set would satisfy $\ell>\underline\ell\geq\ell_{\max}$, contradicting $\ell\leq\ell_{\max}$, so the set is empty. If $\underline\ell<\ell_{\max}$ and $\underline\ell<0$, the fine $\ell=0$ lies in $[0,\ell_{\max}]$ and satisfies $\underline\ell<0<\overline\ell(p)$. If $0\leq\underline\ell<\ell_{\max}$, let $u\triangleq\min\{\ell_{\max},\overline\ell(p)\}$, which exceeds $\underline\ell$ by the display above. The fine $\ell=(\underline\ell+u)/2$ then satisfies $0\leq\underline\ell<\ell<u\leq\ell_{\max}$ and $\ell<\overline\ell(p)$. Thus the set in \eqref{eq:finestrict} is nonempty if and only if $\underline\ell<\ell_{\max}$.

\item \textbf{Part (c).} Suppose $c^0<\underline m$. Then $c^0$ is finite, so $H^0$ is nonempty, and step~1 gives $m(pF)\leq c^0<\underline m$ for every $p\in[0,1]$ and every $F>0$. By step~4, no fine satisfies \eqref{eq:fineweak}. \qedhere
\end{steps}
\end{proof}

\clearpage
\begingroup
\setstretch{1}
\small
\raggedright
\bibliographystyle{aer}
\bibliography{references-v11}
\endgroup

\clearpage
\setcounter{page}{1}
\renewcommand{\thepage}{OA-\arabic{page}}
\setcounter{footnote}{0}
\setcounter{section}{0}
\renewcommand{\thesection}{OA.\arabic{section}}
\renewcommand{\theHsection}{OA.\arabic{section}}
\begin{center}
{\large\bfseries Online Appendix}\\[0.4em]
{When Does Randomized Oversight Align AI Agents That Can Conceal?}\\[0.2em]
{Joshua S. Gans and Richard Holden}
\end{center}

\section{Overlapping Evaluators and Uncertainty}\label{sec:robust}

Evaluators can inspect several channels, and their costs can differ. Index evaluators by $j\in\{1,\ldots,J\}$. For channel $i$, a nonempty compact set $Z_i$ indexes the possible primitive models together with the agent's type and concealment choice. Continuous functions
\[
G_i(z)\triangleq g_i(z)-c_i(z),\qquad Q_{ji}(z)\geq0,\qquad z\in Z_i,
\]
give the net private gain and the expected utility loss conditional on evaluator $j$; for example, $Q_{ji}=F_jd_{ji}$ if evaluator $j$ detects the deviation with probability $d_{ji}$ and imposes loss $F_j$. The uncertainty set is a maintained description of possible models, not a claim that finite observations reveal every concealment strategy. The provider chooses mutually exclusive audit probabilities $p_j\geq0$ with $\sum_jp_j\leq1$, and evaluator $j$ costs $\omega_j>0$. We retain commitment, correct beliefs, and no false sanctions on authorized actions. The least costly policy guaranteeing a utility margin $\varepsilon>0$ for every modeled deviation solves
\begin{equation}\label{eq:LP}
\begin{aligned}
\min_{p\in\R^J_+}\quad & \sum_{j=1}^J\omega_jp_j\\
\text{subject to}\quad & \sum_{j=1}^Jp_j\leq1,\\
& \sum_{j=1}^Jp_jQ_{ji}(z)\geq G_i(z)+\varepsilon
\quad\text{for every }i\text{ and }z\in Z_i.
\end{aligned}
\end{equation}
A policy exists exactly when these constraints are feasible, and the minimum is then attained.\footnote{The incremental utility of deviation $(i,z)$ is $G_i(z)-\sum_jp_jQ_{ji}(z)$. Requiring it to be at most $-\varepsilon$ gives the constraint in \eqref{eq:LP}. An intersection of closed half-spaces with the probability simplex is closed and bounded, hence compact. The continuous cost function attains its minimum on a nonempty feasible set. Mixtures entirely over unauthorized actions preserve the margin; a mixture assigning probability $a>0$ to misconduct is worse by at least $a\varepsilon$.}

Each constraint covers one compound deviation. A review of independently retained records could expose both unauthorized flag acquisition and attempts to disguise it, supplying overlapping coverage. Several evaluators reading the same falsified record instead share a common blind spot. What matters is which deviations their evidence can expose, not how many evaluators are named. With finitely many modeled deviations, \eqref{eq:LP} is a linear program; with a continuum, it has finitely many variables but potentially infinitely many constraints. For separate evaluators, $Q_{ji}=0$ for $j\neq i$ and $Q_{ii}\triangleq Q_i>0$, the constraints reduce to $p_i\geq a_i^\varepsilon\triangleq\max\{0,\max_{z\in Z_i}[G_i(z)+\varepsilon]/Q_i(z)\}$, so the policy is feasible exactly when $\sum_ia_i^\varepsilon\leq1$, with unique cost-minimizing probabilities $p_i=a_i^\varepsilon$.\footnote{Continuity and positivity of $Q_i$ on compact $Z_i$ give a positive minimum, so the ratio is continuous and attains a finite maximum. Dividing each constraint by $Q_i(z)>0$ reduces the channel's constraints to $p_i\geq a_i^\varepsilon$. Their sum yields feasibility. Any other feasible vector increases at least one coordinate without lowering another; since every $\omega_i>0$, its cost is strictly higher.} If all evaluators assign zero loss to a state with $G_i(z)+\varepsilon>0$, the program is infeasible, and more audit expenditure cannot cure that failure. As in Proposition~\ref{prop:unrestricted}, raising a credible $F_j$ with detection fixed weakly expands feasibility and weakly lowers the minimum expenditure, but it cannot reach a state that no evaluator detects.

A design based on estimated incentives needs a margin large enough to cover errors in those estimates. Suppose estimation errors are uniformly bounded over all indices and states by $|G_i-\widehat G_i|\leq\delta_G$ and $|Q_{ji}-\widehat Q_{ji}|\leq\delta_Q$, where $\delta_G,\delta_Q\geq0$. A policy meeting the estimated constraints with margin $\varepsilon$ guarantees a true margin of at least
\begin{equation}\label{eq:transfermargin}
\varepsilon-\delta_G-\delta_Q.
\end{equation}
Thus, a margin exceeding $\delta_G+\delta_Q$ preserves strict deterrence.\footnote{For each deviation, $\sum_jp_jQ_{ji}-G_i\geq\sum_jp_j\widehat Q_{ji}-\widehat G_i-\delta_Q\sum_jp_j-\delta_G\geq\varepsilon-\delta_Q-\delta_G$, using $p_j\geq0$ and $\sum_jp_j\leq1$. The bound is uniform because both primitive-error bounds are uniform. The loss-estimation error is probability-weighted, so its bound does not grow with the number of evaluators when at most one is selected.} The missing provenance check in the Hugging Face incident shows why applying this guarantee also requires verifying implementation. If a design assigns an expected loss $\widehat Q_{ji}>0$ to a check that does not run, its implemented loss is $Q_{ji}=0$, and the error allowance must cover that discrepancy. Uniform bounds must also cover concealment chosen in response to the policy, and accuracy on previously observed, unmodified records does not establish them. These results minimize audit expenditure under fixed incentives, rather than total social loss.

\section{Detection Examples}\label{oa:detection}

These examples illustrate the catch-rate condition \eqref{eq:logcurv} that follows Proposition~\ref{prop:detectionharm}. Both examples use $c(h)=\kappa h$ with $\kappa>0$ and $h\geq0$.

\textit{Log-convex detection.} Let $d(h)=(1+h)^{-2}$. For $t>0$ the objective $\kappa h+t(1+h)^{-2}$ is strictly convex, so for $t>\kappa/2$ its unique minimizer is interior and satisfies $(1+h^*)^3=2t/\kappa$. The catch rate is $q(p)=(t/F)(1+h^*)^{-2}=(\kappa/2)^{2/3}t^{1/3}/F$, which strictly increases with $p$, as \eqref{eq:logcurv} requires since $(\log d)''=2/(1+h)^2>0$.

\textit{Log-concave detection.} Let $d(h)=e^{-\alpha h^2}$ with $\alpha>0$. A positive stationary minimum satisfies $\kappa=2\alpha t h\,e^{-\alpha h^2}$ and $h>1/\sqrt{2\alpha}$; the only other candidate is $h=0$. At a positive stationary point, the difference from the unconcealed cost is $\kappa[2\alpha h^2+1-e^{\alpha h^2}]/(2\alpha h)$. The equation $e^y=1+2y$ has a unique positive solution $y_S\simeq1.25643$; its left side minus its right side decreases until $y=\log2$ and then increases strictly to infinity. Hence positive concealment is the unique global minimum for $t>t_S\triangleq\kappa e^{y_S}/(2\sqrt{\alpha y_S})$, whereas zero is the unique minimum below $t_S$, and both minimize at equality. On the positive branch, $q(p)=\kappa/(2\alpha Fh^*)$ strictly decreases because $h^*$ strictly increases. For $\kappa=0.1$ and $\alpha=F=1$, $t_S\simeq0.15670$; raising $p$ from $0.2$ to $1$ increases optimal concealment from approximately $1.277$ to $1.908$ and lowers the catch rate from approximately $0.0392$ to $0.0262$ (the value at $p=1$ is the continuous endpoint limit of the branch). With gain $g=0.25$, misconduct remains strictly attractive throughout, since even at $p=1$ the minimized cost is approximately $0.2170$; with $g=0.2$, misconduct ceases above approximately $p=0.540$, so the final catch rate would describe only a counterfactual continuing violator.

\section{Risk-Sensitive Scoring and Reward Ensembles}\label{oa:ensemble}

The entropic certainty equivalent in \eqref{eq:CEscore} connects the training objective of Section~\ref{sec:opacity} to conservative scoring rules already studied in AI. A reward-model ensemble is a collection of models that score the same outcome. \citet[eqs.~(4)--(5)]{coste2024} implement and test both the lowest ensemble score and the mean score less a multiple of the variance in reinforcement-learning experiments. The entropic certainty equivalent approaches the former as $\eta$ grows and approximates the latter for small $\eta$, with variance coefficient $\eta/2$. \citet[Sections~3.3--3.5]{hahami2026} establish this connection through a robust objective that gives greater weight to unfavorable score distributions. Risk sensitivity can therefore be a property of the provider's training rule. The resources needed to reach the optima of such a rule are a separate issue in risk-sensitive learning \citep{fei2020}.

The two comparisons follow by factoring out the lowest score and expanding around $\eta=0$. Fix $J\geq1$ finite scores $W_j\in\R$, drawn with probabilities $p_j>0$ that satisfy $\sum_{j=1}^Jp_j=1$. Write $\mathrm{CE}_\eta\triangleq-\eta^{-1}\log\sum_jp_je^{-\eta W_j}$ for $\eta>0$.

\paragraph{Lowest-score limit.} Let $w\triangleq\min_jW_j$ and $P_w\triangleq\sum_{j:W_j=w}p_j>0$. Since $W_j-w\geq0$, the sum $S_\eta\triangleq\sum_jp_je^{-\eta(W_j-w)}$ lies in $[P_w,1]$. Thus $\mathrm{CE}_\eta=w-\eta^{-1}\log S_\eta$ and, because multiplication by $-1/\eta<0$ reverses inequalities,
\[
w\leq\mathrm{CE}_\eta\leq w-\frac{\log P_w}{\eta}\longrightarrow w\quad\text{as }\eta\to\infty.
\]

\paragraph{Small-risk approximation.} Define $\overline W\triangleq\sum_jp_jW_j$ and $\sigma^2\triangleq\sum_jp_j(W_j-\overline W)^2$. The finite sum of exponential Taylor expansions gives $\sum_jp_je^{-\eta W_j}=1-\eta\overline W+\eta^2\sum_jp_jW_j^2/2+O(\eta^3)$. Applying $\log(1+u)=u-u^2/2+O(u^3)$ yields
\[
\log\sum_jp_je^{-\eta W_j}=-\eta\overline W+\frac{\eta^2\sigma^2}{2}+O(\eta^3),\qquad
\mathrm{CE}_\eta=\overline W-\frac{\eta\sigma^2}{2}+O(\eta^2).
\]
These are pointwise statements for fixed scores as $\eta\downarrow0$. If all scores coincide, $P_w=1$, $\sigma^2=0$, and $\mathrm{CE}_\eta=w$ for every $\eta>0$.

\section{A Provider That Values the Weaker Dimension}\label{oa:min}

The finite optimum in Proposition~\ref{prop:objective} reflects the provider's valuation of complementary output. A provider valuing only $\min\{e_1,e_2\}$ would choose $\eta=\infty$ whenever $\gamma>\lambda-1$ and $\tau<4\gamma(\gamma-\lambda+1)/[(1+\lambda)(\lambda-1)^2E]$, because real effort on the expensive dimension then rises with $\eta$ while manipulation is constant.

\section{A Common Disclosure Policy}\label{oa:common}

If the opportunity is privately known, the provider may need a single fine to work across opportunities. For a finite nonempty set of types $\theta$, let $\underline\ell_\theta$ and $\overline\ell_\theta(p)$ be the bounds \eqref{eq:LU} of Proposition~\ref{prop:disclosure} for type $\theta$, with common $p$ and $\ell_{\max}$. Intersecting the type-specific inequalities shows that a common fine exists for weak implementation exactly when
\begin{equation}\label{eq:common}
\max\{0,\max_\theta\underline\ell_\theta\}
\leq\min\{\ell_{\max},\min_\theta\overline\ell_\theta(p)\}.
\end{equation}
For strict implementation, the fine must lie in $[0,\ell_{\max}]\cap(\max_\theta\underline\ell_\theta, \min_\theta\overline\ell_\theta(p))$; finiteness ensures that these extrema are attained. Opportunity-dependent treatment would require observable or separately implementable differences, and the interval does not itself elicit private information. The same reporting rule might have to cover an agent that has copied a flag and one that has acquired access enabling further unauthorized activity. Their retained gains and reasons to continue can differ enough that individually feasible treatments have no common value. For a longer task, the continuation comparison must also be checked at every relevant history with the benefits, costs, and opportunities remaining there.

\end{document}